\documentclass[11pt]{article}

\usepackage{caption}
\usepackage{float} 

\usepackage{authblk}
\usepackage[utf8]{inputenc}
\usepackage[a4paper, total={6.5in, 8in}]{geometry}
\usepackage[english]{babel}
\usepackage{amsmath, amsthm, amssymb}
\usepackage{dsfont}
\usepackage{mathtools}
\usepackage{thmtools}
\usepackage{thm-restate}
\usepackage{mismath}
\usepackage[toc,page]{appendix}
\usepackage{titlesec}
\usepackage{algorithm, algcompatible}
\usepackage[noend]{algpseudocode}
\newlength\myindent
\usepackage[pdfstartview=FitH,colorlinks,linkcolor=blue,filecolor=blue,citecolor=blue,urlcolor=blue]{hyperref} 
\usepackage{graphicx}
\graphicspath{ {./images/} }
\newtheorem{theorem}{Theorem}
\newtheorem{theorempart}{Theorem}[theorem]
\newtheorem{remark}{Remark}
\newtheorem{observation}[theorem]{Observation}
\newtheorem{lemma}{Lemma}
\def \RR   {{\mathbb R}}

\usepackage{float}
\usepackage{tcolorbox}
\usepackage{amsthm}
\usepackage{mathtools}
\usepackage{amsfonts}
\usepackage{dsfont}
\usepackage{graphics}
\usepackage{xcolor}
\usepackage{dsfont}
\usepackage{thmtools}
\usepackage{thm-restate}
\usepackage{authblk}
\usepackage{hyperref}

\title{Load Balancing with Duration Predictions}

\author{
  Yossi Azar\thanks{Department of Computer Science,
  Tel Aviv University, Israel. \href{mailto:azar@tauex.tau.ac.il}{azar@tauex.tau.ac.il}.}
  \quad
  Niv Buchbinder\thanks{Department of Statistics and Operations Research, School of Mathematical Sciences,
  Tel Aviv University, Israel. \href{mailto:niv.buchbinder@gmail.com}{niv.buchbinder@gmail.com}.
  The work of Niv Buchbinder is supported in part by the Israel Science Foundation (ISF) grant no.\ 3001/24, 
  and the United States–Israel Binational Science Foundation (BSF) grant no.\ 2022418.}
  \quad
  Tomer Epshtein\thanks{Department of Computer Science,
  Tel Aviv University, Israel. \href{mailto:tomere1@mail.tau.ac.il}{tomere1@mail.tau.ac.il}.}
}

\date{}

\begin{document}

\maketitle

\thispagestyle{empty}
\pagenumbering{Alph}

\begin{abstract}

We study the classic fully dynamic load balancing problem on unrelated machines where jobs arrive and depart over time and the goal is minimizing the maximum load, or more generally the $\ell_p$-norm of the load vector. 
Previous work either studied the {\em clairvoyant} setting in which exact durations are known to the algorithm, or the {\em unknown duration} setting in which no information on the duration is given to the algorithm. For the clairvoyant setting algorithms with polylogarithmic competitive ratios were designed, while for the unknown duration setting strong lower bounds exist and only polynomial competitive factors are possible. 

We bridge this gap by studying a more realistic model in which some estimate/prediction of the duration is available to the algorithm. We observe that directly incorporating predictions into classical load balancing algorithms designed for the clairvoyant setting can lead to a notable decline in performance.
We design better algorithms whose performance depends smoothly on the accuracy of the available prediction. We also prove lower bounds on the competitiveness of algorithms that use such inaccurate predictions. 

\end{abstract}

\newpage
\pagenumbering{arabic}

\section{Introduction}

Load balancing is a fundamental problem studied extensively in both computer science and operations research. In the basic offline setting of scheduling permanent jobs on unrelated machines, we are given a set of $m$ machines and a sequence $\sigma$ of $n$ permanent jobs. The processing load of job $j$ on machine $i$ is $p_{ij}\in \RR_+ \cup \{\infty\}$ (in the restricted assignment model $p_{ij}=p_j$ for a subset of machines). The goal is to minimize the maximum load, or more generally the $\ell_p$-norm of the load vector of the machines. In the online setting jobs arrive one-by-one and should be assigned to a machine irrevocably upon arrival.

In a more realistic model jobs may arrive and also depart over time. In this model each job $j$ has an arrival time $t_j$ and a duration $d_j$, i.e. the job stays in the system starting at $t_j$ until $t_j + d_j$. Previous work can be partitioned to the {\em clairvoyant setting} in which the duration of each job is given to the algorithm upon arrival of the job, and the {\em non-clairvoyant} setting in which the duration is unknown to the algorithm. For the clairvoyant setting algorithms with polylogarithmic competitive ratios were designed. For example, \cite{AKPPW97} designed an $\bigO(\log m + \log D)$- competitive algorithm for the maximum load objective, where $D$ is ratio between the maximum and the minimum duration of a job. This result matches up to logarithmic factors the known lower bound of $\Omega\left(\log m + \sqrt{\frac{\log D}{\log \log D}}\right)$~\cite{ANR95,AAE03}.
The non-clairvoyant setting is much harder (in terms of competitive ratio). For unrelated machines a lower bound of $\Omega\left(\frac{m}{\log m}\right)$ was given in \cite{AAE03} and even for the restricted assignment model tight bounds of $\Theta(\sqrt{m})$ on the competitive ratio are known~\cite{ABK94,AKPPW97}.
More results on temporary jobs for both the $\ell_{\infty}$-norm and the $\ell_p$-norm appear in \cite{AE04, AEE06, MP97}.

On one hand, a polynomial in $m$ competitive factor for the non-clairvoyant setting is quite poor and undesirable. However, on the other hand, a fair criticism is that the clairvoyant model in which accurate durations of the jobs are known in advance is extremely unrealistic. As a compromise, a more realistic assumption is that some estimate/prediction of the duration of a job is available to the algorithm. Such estimates on the durations may be available via machine learning and statistical models (See, e.g., \cite{YN21}).
A natural question is whether we can design an algorithm that can achieve good competitive ratio when given such inaccurate estimates. This question has been explored in various scheduling problems with different objective functions~\cite{LLMV20}. For example, Azar, Leonardi and Touitou as well as Zhao, Li, Zomaya studied recently the problem of minimizing the flow time with uncertain processing times~\cite{ALT21,ALT22,ZLZ22}.
The general question of designing ``robust" algorithms when given predictions was studied extensively in recent years in countless other scenarios (See also Subection \ref{sec:related}). 

\subsection{Our Results}

We study a model in which jobs arrive online at time steps $t=0,1,2, \ldots,T-1$. Each job j has a real duration $d_j \in [1, D]$, but upon $j$'s arrival the algorithm is given only an estimation $1 \leq \Tilde{d_{j}} \leq \tilde{D}$ (may not be an integer) on its duration such that $d_j \in [\lceil \Tilde{d_{j}}/\mu_{2} \rceil, \lfloor \mu_{1} \cdot \Tilde{d_{j}} \rfloor]$. We refer to $\mu_1, \mu_2 \geq 1$ as the maximum underestimation and maximum overestimation factors of the durations and let $\mu \triangleq \mu_{1}\cdot \mu_{2}$ be the distortion parameter of the input\footnote{These parameters do not need to be known in advance.}. Let $\ell(t) \triangleq (\ell_{1}(t),\ldots,\ell_{m}(t))$ be the load vector at time $t$, where $\ell_i(t)$ is the load on machine $i$ at time $t$.
We study the problem of minimizing the maximum load
$\max_{t=1}^{T}\{\|\ell(t)\|_{\infty}\}$, where $\|\ell(t)\|_{\infty} = \max_{i=1}^{m} \ell_{i}(t)$, or ,more generally, minimizing $\max_{t=1}^{T}\{\|\ell(t)\|_{p}\}$, where   $\|\ell(t)\|_{p} = (\sum_{i=1}^{m} \ell_{i}(t)^{p})^{\frac{1}{p}}$ for $p\geq 1$\footnote{Without loss of generality, in our upper and lower bounds, we make the assumption that $1 \leq p \leq \log m$ as for any $m$-dimensional vector $x$, $\|x\|_{\infty} = \Theta(\|x\|_{\log m})$. Therefore, for any $p > \log m$, an algorithm can instead use an $\ell_{\log m}$-norm to approximate an $\ell_p$-norm objective up to a multiplicative constant.}.

An intuitive first attempt to handle such inaccurate estimates is to apply existing algorithms designed for the clairvoyant scenario using the estimates $\Tilde{d_{j}}$ instead of the actual durations $d_j$.
However, we show that for any fixed $\mu$ such a naive approach yields an $\Omega(\log \tilde{D} \cdot \log m)$-competitive algorithm. Our main result is a $\bigO(\log \tilde{D} + \log m)$-competitive algorithm improving ''quadartialy" the bound of the naive algorithm. Specifically:

\begin{theorem}
\label{thm:upper}
    For any $p \geq 1$ there exists an $\bigO\left(\mu (p+\log (\mu \tilde{D}))\right)$-competitive algorithm for load balancing of temporary tasks with predictions for the $\ell_p$-norm objective. In particular, for $\ell_{\infty}$, there exists an $\bigO\left(\mu\log (m\mu \tilde{D})\right)$-competitive algorithm.
\end{theorem}

We complement the upper bound by a lower bound, showing that the logarithmic factors and a polynomial dependency on $\mu$ are indeed necessary.

\begin{theorem}
\label{thm:lower}
    Let $p\in [1,\log m]$, $\mu\leq m$, $\tilde{D}\leq {m}^{m}$. Then, any deterministic online algorithm for the load balancing of temporary tasks has a competitive ratio of at least \\
    $\Omega\left(\mu^{\frac{1}{2} - \theta(\frac{1}{p})} + p+(\frac{\log \tilde{D}}{\log \log \tilde{D}})^{\frac{1}{2}-\theta(\frac{1}{p})}\right)$, where $\frac{1}{2}-\theta(\frac{1}{p})=\frac{p}{2p-1}-\frac{1}{p}$.
    In particular, for $\ell_{\infty}$, the competitive ratio is at least
    \[
    \Omega\left(\sqrt{\mu}+\log m+\sqrt{\frac{\log \tilde{D}}{\log \log \tilde{D}}} \right). 
    \]
    This lower bounds holds even for the restricted assignment model.
\end{theorem}

We also study a useful concept that we refer to as the {\em Price of Estimation} (PoE). Let $(I,A)$ be an instance of the problem and an assignment of jobs to machines, and let $(I(\mu),A)$ be the same instance and assignment where each job $j$ has duration $d'_j=\mu \cdot d_j$. 
Let $\ell(I,A, t)$ and $\ell(I(\mu),A, t)$ be the load vector at time $t$ when using the job assignment $A$ for instance $I$ and $I(\mu)$ respectively. Then,
\begin{equation}
    PoE(\mu,p) = \sup_{(I,A)}\left\{\frac{\max_{t=1}^{T}\|\ell(I(\mu),A ,t)\|_p}{\max_{t=1}^{T}\|\ell(I, A, t)\|_p}\right\} \nonumber
\end{equation}
Intuitively, the price of estimation captures the maximum penalty in our objective function (of any instance and assignment)  due to a distortion of $\mu$ in the jobs' durations. In particular, if we choose $A$ to be the optimal assignment the price of estimation bounds the maximum increase in the optimal value that can happen due to such distortion. We prove the following.

\begin{theorem}
\label{thm:PoE}
 For any $p\geq 1$, the price of estimation $PoE(\mu,p) = \Theta(\mu \log \tilde{D})$. Specifically, 
\begin{itemize}
    \item For any instance $I$ and an assignment $A$, \\ $\max_{t=1}^{T}\|\ell(I(\mu),A ,t)\|_p= \bigO(\mu \log \tilde{D}) \cdot\max_{t=1}^{T}\|\ell(I, A,t)\|_p$.
    \item There exists an instance and an assignment $A$ for which $$\max_{t=1}^{T}\|\ell(I(\mu),A,t)\|_p = \Omega(\mu \log \tilde{D}) \cdot \max_{t=1}^{T} \|\ell(I,A,t)\|_p .$$
\end{itemize}
\end{theorem}

Finally, we show that our upper bounds can be extended for a setting of the routing problem with of temporary jobs.
In the routing problem, we are given a graph $G=(V,E)$ and a sequence of incoming jobs. Each job $j$ is associated with $j=(s_j, t_j, \{p_{e,j}\}_{e \in E},\tilde{d_j})$, where $s_j$ is the source node, $t_j$ is the target node, $p_{e,j}\in\{\RR_+, \infty\}$ is the load added to  edge $e \in E$ if the job is routed through the edge, and $\tilde{d}_j$ is to the estimated duration of the job.
Upon the arrival of job $j$, the algorithm should establish a connection by choosing a path between $s_j$ and $t_j$ and allocating the necessary bandwidth along the selected path.
 
\subsection{Our Techniques}

Our algorithms build on ideas developed earlier for the clairvoyant case in which accurate durations are available to the algorithm ~\cite{AKPPW97}.
A first attempt to our problem settings may be to execute the same algorithm with the load vector based on the estimated durations.
However, we show that such a naive approach yields an $\Theta(\mu \log \tilde{D} \cdot \log m)$-competitive algorithm.

Our algorithms are based on the greedy approach, and the most relevant work is the algorithm designed in \cite{IKKP19} for the setting of online vector scheduling problem.
When adapted to our setting, this algorithm reduces to the following.
When a new job arrives, the algorithm assigns the job to a machine $i$ that minimizes $f(\ell(t,i))-f(\ell(t))$, where $\ell(t)$ is the load vector just before assignment, $\ell(t,i)$ is the load vector after hypothetically assigning the job to machine $i$, and $f$ is a carefully chosen function of the load vector.

Instead of using the estimates $\tilde{d_j}$ directly, we take a slightly different approach. Let $d'_j= \mu_1 \cdot \tilde{d_j}$ be the maximum overestimate of the job duration. Based on these overestimate values, we define a pseudo load vector $\tilde{\ell}$ on the machines. The pseudo load does not change, even if some jobs end earlier than $t_j+d'_j$. Our algorithm then executes the algorithm in \cite{IKKP19} with these pseudo loads.

In principle, one might expect, that our algorithm loses a multiplicative factor of $O(\mu\log \tilde{D})$ due to the price of estimation discussed earlier. Surprisingly, we show that this is not the case as long as we use both $d'$ and the pseudo loads. For our improved analysis, we utilize a simple but important fact that upon arrival of a job, the estimated future pseudo loads on the machines are monotonically decreasing. That is, as expected, the final vector $\ell$ (or $\tilde{\ell}$) may be arbitrary increasing or decreasing over time, meaning that $\ell_i(t)$ may be smaller or larger than $\ell_i(t+1)$ for a machine $i$. However, when considering the jobs known to the algorithm until time $t$ without future arriving jobs (i.e., taking a snapshot at time $t$) the load vector is monotonically decreasing over time. Our analysis crucially uses this simple fact to obtain our improved competitive ratio.

\subsection{Related Work}\label{sec:related} 

Online load balancing of permanent jobs (jobs that never depart) with the goal of minimizing makespan has been extensively studied~\cite{G69,KVV90,ANR95,BFKV95,KPT96,AAFPW97,AAPW01,GKS14,BEM22}. Graham initiated the study of the {\em identical machines} model \cite{G69} which was further explored in \cite{BFKV95,KPT96}.
Later, $O(\log m)$-competitive algorithms were designed for the restricted assignment model~\cite{ANR95} and the unrelated machines model~\cite{AAFPW97} along with a matching lower bound~\cite{ANR95}. The more general objective of minimizing the $\ell_p$-norm of the load vector was also explored~\cite{CW75,CC76,AAGKKV95}. For this more general objective Caragiannis \cite{C08} proved that the greedy algorithm is $\frac{1}{2^{1/p}-1} \approx 1.4427p$-competitive for unrelated machines and that no deterministic algorithm can achieve a competitiveness of $\frac{1}{2^{1/p}-1} - \epsilon$ for any $\epsilon > 0$.

A further generalization of the load balancing problem known as {\em vector scheduling} was introduced by Chekuri and Khanna~\cite{CK04}. In this problem each job is associated with a load vector of dimension $d$. The goal is to minimize the maximum machine load across all dimensions ($\ell_{\infty}$) or, more generally, minimize the maximum value of the $\ell_p$-norm of the loads of the machines on each dimension.
Meyerson, Roytman, and Tagiku \cite{MRT13} designed a $\bigO(\log m + \log d)$-competitive algorithm for the case of unrelated machines and $\ell_{\infty}$-norm. Later on, Im et al. \cite{IKKP19} designed a $\Theta(p + \log d)$-competitive algorithm for the $\ell_p$-norm case, and proved a matching lower bound on the competitive ratio. 

Recently, motivated by machine learning techniques, a setting in which the online algorithm is given an additional (possibly inaccurate) predictive data has been studied extensively. The goal is to design algorithms that have a good performance when given an accurate predictions, while performing reasonably well when predictions are inaccurate.
In the domain of the load balancing problem, current research focuses on predicting job processing times using machine learning and statistical models~\cite{YN21}. Online algorithms utilizing estimations for job duration have been explored in many settings \cite{LLMV20,ALT21,ALT22,ZLZ22}.

\section{Preliminaries}\label{sec:preliminaries}

We study the problem of assigning temporary tasks to unrelated machines with time predictions. In this problem we are given $m$ machines. Jobs arrive online at time steps $t=0,1,2, \ldots,T-1$. Each job $j$ is associated with an arrival time $t_j$, a duration $d_j\in \{1,2, \ldots, D\}$, and a processing load $p_{ij}\in\{\RR_+, \infty\}$ when it is assigned to machine $i$. At any time $t$ the jobs arrive in an arbitrary order.
Note that the well-studied permanent jobs is a special case when all the jobs arrive at time 0.
The $t$-th time slot corresponds to time $(t-1, t)$. In particular job $j$ which arrives at $t_j$ and departs at $t_j + d_j$ is alive at time slots $t_j+1,\ldots,t_j+d_j$.
We assume that upon the arrival of a job the algorithm is given a (possibly inaccurate) duration estimation $\tilde{d}_j$ for its duration. We define $\tilde{D}$ as the maximum estimation over all the jobs.
Let $\mu_1, \mu_2 \geq 1$ be maximum underestimation and maximum overestimation factors of the durations such that for all jobs $d_j \in [\lceil \frac{\Tilde{d_{j}}}{\mu_{2}} \rceil, \lfloor \mu_{1} \cdot \Tilde{d_{j}} \rfloor]$ and let $\mu \triangleq \mu_{1}\cdot \mu_{2}$ be the distortion parameter of the input. Note that that the parameters $\mu_1,\mu_2, D, \tilde{D}$ do not need to be known to the algorithm in advance.

Recall that $\ell_i(t)$ is the load on machine $i$ at time $t$.
Let $\ell(t) \triangleq (\ell_{1}(t),\ldots,\ell_{m}(t))$ be the load vector at time slot $t$. 
Let $y_{i,j}$ be an indicator that the online algorithm assigns job $j$ to machine $i$. The load of machine $i$ at time slot $t$ after the arrival of the $j$-th job is $\ell_{i,j}(t) \triangleq \sum_{k=1,\ldots, j,  t\in [t_k+1, t_k+d_k]}p_{i,k}\cdot y_{i,k}$. 
Correspondingly, we define $y^*_{i,j}$, $\ell^*_{i,j}(t)$, and $\ell^*(t)$ as the assignment and load vector of the optimal solution.
Finally, let $\tilde{\ell}_{i,j}(t) \triangleq \sum_{k=1,\ldots, j, t\in [t_k+1, t_k+\lfloor \mu_1 \cdot \Tilde{d_k} \rfloor]}p_{i,k}\cdot y_{i,k}$ to be the ''overestimate" load on machine $i$ at time $j$ after the $j$th job arrive ,and $\Tilde{\ell}(t) \triangleq (\Tilde{\ell}_{1}(t),\ldots,\Tilde{\ell}_{m}(t))$ be the ''overestimate" load vector.
Let $\|\ell(t)\|_{p} = (\sum_{i=1}^{m} \ell_{i}(t)^{p})^{\frac{1}{p}}$, and $\|\ell(t)\|_{\infty} = \max_{i=1}^{m} \ell_{i}(t)$. We study the objective of minimizing $\max_{t=1}^{T}\{\|\ell(t)\|_{p}\}$.

We also study a useful concept that we refer to as the {\em Price of Estimation} (PoE).
Let $(I,A)$ be an instance of the problem and an assignment of jobs to machines, and let $(I(\mu),A)$ be the same instance and assignment where each job $j$ has duration $d'_j=\mu \cdot d_j$.
Let $\ell(I,A, t)$ and $\ell(I(\mu),A, t)$ be the load vector at time $t$ when using the job assignment $A$ for instance $I$ and $I(\mu)$ respectively. Then,
\begin{equation}
    PoE(\mu,p) \triangleq \sup_{(I,A)}\left\{\frac{\max_{t=1}^{T}\|\ell(I(\mu),A ,t)\|_p}{\max_{t=1}^{T}\|\ell(I, A, t)\|_p}\right\} \nonumber
\end{equation}

\noindent For simplicity, for a single machine, we use the notation $PoE_1(\mu,p)$ instead of $PoE(\mu, p)$. However, both notations refer to the same quantity. We note that for any $p$:
\begin{equation}
    PoE_1(\mu) = PoE_1(\mu,p)  = \sup_{I}\left\{\frac{\max_{t=1}^{T}\ell(I(
    \mu),t)}{\max_{t=1}^{T}\ell(I,t)}\right\}\nonumber
\end{equation}
Where in this case $\ell(I, t)$ represents the load on the single machine at time $t$ induced by $I$.

For simplicity of notation we also define for any instance $I$ over $m$ machines and an assignment $A$,
\begin{align*}
    I_{i}(I, A) \triangleq \{j \in I | A(j) = i\}
\end{align*}
That is, $I_{i}(I, A)$ is the set of jobs in $I$ which are assigned to machine $i$ by $A$.
In addition, we define $\ell_i(I_i, t)$ to be the load on machine $i$ by time $t$ when the jobs in instance $I_i$ are assigned to machine $i$ (thus, no assignment is labeled in the definition and we only observe on jobs instance).
Using this notation, $\ell(I(\mu),A, t)=(\ell_1(I_1, t), \ell_2(I_2, t)\ldots, \ell_m(I_m, t))$ .
\setcounter{observation}{0}
\begin{observation}\label{POE=POE_1}
    For every $\mu, p \geq 1$, $PoE_1(\mu) = PoE(\mu, p)$.
\end{observation}
\begin{proof}
    It is clear that $PoE_1(\mu) \leq PoE(\mu, p)$ since taking an instance of jobs $I$ with single machine is a unique instance for $PoE(\mu, p)$'s definition  (no assignment $A$ is required as there's only a single machine, mandating that all jobs in the instance must be assigned to it).

    Consider $I$ and $A$ as an instance and assignment of jobs over $m$ machines, respectively, and let $t \in [1, T]$ represent a specific time. According to the definition of $PoE_1(\mu)$, for each machine $1 \leq i \leq m$, the load induced on machine $i$ by time $t$ through $I_i(I, A)(\mu)$ is at most $PoE_1(\mu)$ times the load on machine $i$ induced by $I_i(I, A)$, expressed as:
    $\ell_i(I_i(I, A)(\mu), t) \leq PoE_1(\mu) \cdot \ell_i(I_i(I,A),t)$.
    Hence, by observing the load vector $\ell(I(\mu),A ,t)$ we conclude that,
    \begin{align*}
        \|\ell(I(\mu),A ,t)\|_p \leq PoE_1(\mu) \cdot \| \ell(I, A, t) \|_p \leq PoE_1(\mu) \cdot \max_{t=1}^{T} \| \ell(I, A, t) \|_p
    \end{align*}
    Since this is true for every $t$, we get, 
    \begin{equation}
    PoE(\mu,p) = \sup_{(I,A)}\left\{\frac{\max_{t=1}^{T}\|\ell(I(\mu),A ,t)\|_p}{\max_{t=1}^{T}\|\ell(I, A, t)\|_p}\right\} \leq PoE_1(\mu) \nonumber
    \end{equation}
    
\end{proof}

In this paper, we also tackle the routing problem, which involves a given graph $G=(V,E)$ and a sequence of incoming jobs. Jobs arrive online at time steps $t=0,1,2, \ldots,T-1$. Each job $j$ arrives at time $a_j$ and is associated with $j=(s_j, t_j, \{p_{e,j}\}_{e \in E},\tilde{d_j})$, where $s_j$ represents the source node of the job, $t_j$ is the target node, $p_{e,j}\in\{\RR_+, \infty\}$ denotes the load added to each edge $e \in E$ if the job is routed through that edge, and $\tilde{d_j}$ and $d_j\in \{1,2, \ldots, D\}$ correspond to the estimated and actual duration times of the job, respectively.
Upon the immediate arrival of job $j$, the algorithm establishes a connection by allocating the necessary bandwidth along a selected path between nodes $s_j$ and $t_j$. In other words, the load is increased on all edges within the chosen path.
When a connection concludes, which happens by time $a_j + d_j$, the allocated load is released. The allocation cost is defined as the maximum $\ell_p$-norm of the loads on all edges in the graph, calculated over all points in time.

Let $y_{e,j}$ be an indicator that the online algorithm assigns job $j$ to a route $W_j$ s.t $e \in W_j$. The load of edge $e$ at time slot $t$ after the arrival of the $j$-th job is $\ell_{e,j}(t) \triangleq \sum_{k=1,\ldots, j,  t\in [t_k+1, t_k+d_k]}p_{e,k}\cdot y_{e,k}$. Let $\ell(t) \triangleq (\ell_{1}(t),\ldots,\ell_{|E|}(t))$ be the load vector at time slot $t$. Correspondingly, we define $y^*_{e,j}$, $\ell^*_{e,j}(t)$, and $\ell^*(t)$ as the assignment and load vector of the optimal solution. 
Finally, let $\tilde{\ell}_{e,j}(t) \triangleq \sum_{k=1,\ldots, j, t\in [t_k+1, t_k+\lfloor \mu_1 \cdot \Tilde{d_k} \rfloor]}p_{e,k}\cdot y_{e,k}$ to be the ''overestimate" load on edge $e$ at time $j$ after the $j$th job arrives ,and $\Tilde{\ell}(t) \triangleq (\Tilde{\ell}_{1}(t),\ldots,\Tilde{\ell}_{|E|}(t))$ be the ''overestimate" load vector.
Let $\|\ell(t)\|_{p} = (\sum_{e \in E} \ell_{e}(t)^{p})^{\frac{1}{p}}$, and $\|\ell(t)\|_{\infty} = \max_{e \in E} \ell_{e}(t)$. We study the objective of minimizing $\max_{t=1}^{T}\{\|\ell(t)\|_{p}\}$.

\section{Upper Bounds}

In this Section we prove the Theorem \ref{thm:upper} that we repeat here for convenience.
\setcounter{theorem}{0}
\begin{theorem}
    For any $p \geq 1$ there exists an $\bigO\left(\mu (p+\log (\mu \tilde{D}))\right)$-competitive algorithm for load balancing of temporary tasks with predictions for the $\ell_p$-norm objective. In particular, for $\ell_{\infty}$, there exists an $\bigO\left(\mu\log (m\mu \tilde{D})\right)$-competitive algorithm, which we show in Section \ref{linf_norm_algorithm}.
\end{theorem}

\subsection{General \texorpdfstring{$\ell_p$-norm}{}}\label{sec:lp_norm_upper}

In this Subsection, we design an online algorithm for the load balancing of temporary tasks with predictions in the $\ell_p$-norm, such that for every input series $\sigma$ with distortion $\mu$ it holds:
\[\max_{t}\|\ell(t)\|_p \leq \bigO\left(\mu (p+ \log (\mu \tilde{D}))\right) \cdot \max_{t}\|\ell^*(t)\|_p \]
First and foremost, our approach relies on several assumptions and employs a very simple algorithm that is built upon these assumptions. However, in Subsection \ref{remove-assumptions}, we will demonstrate how to eliminate each assumption while preserving the desired competitive ratio.
The assumptions are the following:
\begin{enumerate}\label{all-assumptions}
    \item\label{series_distortion_params} The parameters $\mu_{1}$ and $\tilde{D}$ are known to the algorithm in advance.
    \item\label{series_end_time} All jobs arrive in the time range 0 up to $\mu_1 \tilde{D}$.
\end{enumerate}

\paragraph*{The Load Balancing Algorithm}\label{lp_norm_algorithm} Let $\Tilde{T} \triangleq 2\mu_1 \tilde{D}$. Upon arrival of a new job $j$ the algorithm allocates the job to the machine $i$ such that:
\begin{align}
    i =  \arg \min_{i'} & \sum_{t=t_j + 1}^{t_j+\lfloor \mu_{1} \cdot \Tilde{d_{j}} \rfloor} \left[(\sum_{z=1}^{m} (\Tilde{\ell}_{z,j-1}(t) + p_{z,j}\cdot \mathds{1}_{z=i'})^{p})^{\frac{p+\log \Tilde{T}}{p}}  - (\sum_{z=1}^{m} \Tilde{\ell}_{z,j-1}(t)^{p})^{\frac{p+\log \Tilde{T}}{p}}\right] \nonumber
\end{align}

Note that this summation involves adding up $\lfloor \mu_{1} \cdot \Tilde{d_{j}} \rfloor$ terms. Each of these terms represents a difference between two values, where the second value is a constant (for all machines $i'$). We choose to retain this term as it has no impact on the algorithm's behavior and proves beneficial for analysis purposes. We also note that due to assumption \ref{series_end_time}, each job $j$ departs by time $t_{j} + d_{j} \leq t_j+\mu_1 \cdot \tilde{d_j} \leq \mu_1 \tilde{D} + \mu_1 \tilde{D} = \Tilde{T}$. Therefore, starting from time $\Tilde{T}$, there are no active jobs remaining.

We prove that the algorithm is $\bigO\left(\mu (p+\log (\mu \tilde{D}))\right)$-competitive for any $1 \leq p \leq \log m$.

\begin{proof} 
For a given $j$, let $\Tilde{L}_{j-1} =  (\sum_{i=1}^{m} \Tilde{\ell}_{i,j-1}(t)^{p})^{\frac{p+\log \Tilde{T}}{p}}$.
We have for each $j$:
{\allowdisplaybreaks
\begin{align}
    &\sum_{t=1}^{\Tilde{T}}\left[ (\sum_{i=1}^{m} \Tilde{\ell}_{i,j}(t)^{p})^{\frac{p+\log \Tilde{T}}{p}} - (\sum_{i=1}^{m} \Tilde{\ell}_{i,j-1}(t)^{p})^{\frac{p+\log \Tilde{T}}{p}} \right]\nonumber\\
     &= \sum_{t=t_j + 1}^{t_j+\lfloor \mu_{1} \cdot \Tilde{d_{j}} \rfloor} \left[(\sum_{i=1}^{m} \Tilde{\ell}_{i,j}(t)^{p})^{\frac{p+\log \Tilde{T}}{p}} -  \Tilde{L}_{j-1} \right]\nonumber\\
    &= \sum_{t=t_j + 1}^{t_j+\lfloor \mu_{1} \cdot \Tilde{d_{j}} \rfloor} \left[(\sum_{i=1}^{m} (\Tilde{\ell}_{i,j-1}(t) + p_{i,j}\cdot y_{i,j})^{p})^{\frac{p+\log \Tilde{T}}{p}}  - \Tilde{L}_{j-1}\right] \label{ineq11}\\
    &\leq \sum_{t=t_j + 1}^{t_j+\lfloor \mu_{1} \cdot \Tilde{d_{j}} \rfloor} \left[(\sum_{i=1}^{m} (\Tilde{\ell}_{i,j-1}(t) + p_{i,j}\cdot y^{*}_{i,j})^{p})^{\frac{p+\log \Tilde{T}}{p}} - \Tilde{L}_{j-1}\right] \label{ineq12}\\
    &\leq \mu \cdot \sum_{t=t_j + 1}^{t_j+ \lceil \frac{\mu_1 \cdot \tilde{d_j}}{\mu} \rceil } \left[(\sum_{i=1}^{m} (\Tilde{\ell}_{i,j-1}(t) + p_{i,j}\cdot y^{*}_{i,j})^{p})^{\frac{p+\log \Tilde{T}}{p}} - \Tilde{L}_{j-1}\right] \label{ineq16}\\
    &\leq \mu \cdot \sum_{t=t_j + 1}^{t_j+ d_{j}} \left[(\sum_{i=1}^{m} (\Tilde{\ell}_{i,j-1}(t) + p_{i,j}\cdot y^{*}_{i,j})^{p})^{\frac{p+\log \Tilde{T}}{p}} - \Tilde{L}_{j-1}\right]\label{ineq13}\\
    &\leq \mu \cdot \sum_{t=t_j + 1}^{t_j+ d_{j}} \left[(\sum_{i=1}^{m} (\Tilde{\ell}_{i}(t) + \ell^{*}_{i,j-1}(t) + p_{i,j}\cdot y^{*}_{i,j})^{p})^{\frac{p+\log \Tilde{T}}{p}}  - (\sum_{i=1}^{m} (\Tilde{\ell}_{i}(t) + \ell^{*}_{i,j-1}(t))^{p})^{\frac{p+\log \Tilde{T}}{p}}\right] \label{ineq14}\\
    &= \mu  \sum_{t=1}^{\Tilde{T}} \left[(\sum_{i=1}^{m} (\Tilde{\ell}_{i}(t) + \ell^{*}_{i,j}(t))^{p})^{\frac{p+\log \Tilde{T}}{p}}   - (\sum_{i=1}^{m} (\Tilde{\ell}_{i}(t) + \ell^{*}_{i,j-1}(t))^{p})^{\frac{p+\log \Tilde{T}}{p}}\right] \nonumber
\end{align}}

\noindent Equality \eqref{ineq11} follows by the definition of $y_{i,j}$. 
Inequality \eqref{ineq12} follows by the greediness of the algorithm.
Inequality \eqref{ineq16} follows since for any $i,j$, $\tilde{\ell}_{i,j}(t) \geq  \tilde{\ell}_{i,j}(t+1)$ and therefore $\mu$ times the sum of the first $\lceil \frac{\lfloor \mu_1 \cdot\tilde{d}_j \rfloor}{\mu} \rceil$ terms is at least as the sum of the $\lfloor \mu_1 \cdot\tilde{d}_j \rfloor$ terms (and $\frac{\lfloor \mu_1 \cdot\tilde{d}_j \rfloor}{\mu} \leq \frac{\mu_1 \cdot\tilde{d}_j}{\mu}$).
Inequality \eqref{ineq13} follows since $\lceil \frac{\mu_1 \cdot \tilde{d_j}}{\mu} \rceil = \lceil \frac{\tilde{d_j}}{\mu_{2}} \rceil \leq  d_{j}$.
Inequality \eqref{ineq14} follows since the function $f(x_{1}, x_{2}, \ldots, x_{m}) = (\sum_{i=1}^{m} (x_{i} + a_{i})^{w})^{z} - (\sum_{i=1}^{m} x_{i}^{w})^{z}$ is a non-decreasing function of $x_i$ for  $w \geq 1$, $z \geq 1$, $a_{i} \geq 0$ and for every $i$ it holds $\tilde{\ell}_{i,j-1}(t) \leq \Tilde{\ell}_{i}(t) \leq \Tilde{\ell}_{i}(t) + \ell^{*}_{i,j-1}(t)$. Summing the above inequality for all $j$ we get,
\begin{align*} 
& \sum_{t=1}^{\Tilde{T}} \|\Tilde{\ell}(t)\|_p^{p+\log \Tilde{T}} =\sum_{t=1}^{\Tilde{T}} (\sum_{i=1}^{m} \Tilde{\ell}_{i}(t)^{p})^{\frac{p+\log \Tilde{T}}{p}}  \leq  \mu \cdot \sum_{t=1}^{\Tilde{T}} \left[(\sum_{i=1}^{m} (\Tilde{\ell}_{i}(t) + \ell^{*}_{i}(t))^{p})^{\frac{p+\log \Tilde{T}}{p}} - (\sum_{i=1}^{m} (\Tilde{\ell}_{i}(t))^{p})^{\frac{p+\log \Tilde{T}}{p}}\right] \\
&= \mu \cdot \sum_{t=1}^{\Tilde{T}} \left[\| \tilde{\ell}(t) + \ell^{*}(t)\|_p^{p + \log \Tilde{T}} - \|\Tilde{\ell}(t)\|_p^{p+\log \Tilde{T}}\right]
\end{align*}
Therefore,
{\allowdisplaybreaks
\begin{align}
\lefteqn{
     (1+\mu) \cdot \sum_{t=1}^{\Tilde{T}} \|\Tilde{\ell}(t)\|_p^{p+\log \Tilde{T}}
     \leq \mu \cdot \sum_{t=1}^{\Tilde{T}} \| \tilde{\ell}(t) + \ell^{*}(t)\|_p^{p + \log \Tilde{T}} } \nonumber \\ 
     &\leq \mu \cdot \sum_{t=1}^{\Tilde{T}} \left[\|\Tilde{\ell}(t)\|_p + \|\ell^{*}(t)\|_p\right]^{p + \log \Tilde{T}} \label{ineq21}
     \\ &\leq \mu \cdot \sum_{t=1}^{\Tilde{T}} \left[(1+\frac{1}{4\mu(p+\log \Tilde{T})})^{p + \log \Tilde{T}} \|\Tilde{\ell}(t)\|_p^{p + \log \Tilde{T}} + (1+ 4\mu(p+\log \Tilde{T}))^{p+\log \Tilde{T}}\|\ell^{*}(t)\|_p^{p + \log \Tilde{T}}\right] \label{ineq22}
     \\ &\leq \mu \cdot \sum_{t=1}^{\Tilde{T}} \left[(1+\frac{1}{2\mu})\cdot \|\Tilde{\ell}(t)\|_p^{p + \log \Tilde{T}}  + (5\mu\cdot(p+\log \Tilde{T}))^{p+\log \Tilde{T}}\cdot \|\ell^{*}(t)\|_p^{p + \log \Tilde{T}}\right] \label{ineq23}
     \\ &= (\mu + \frac{1}{2})\cdot \sum_{t=1}^{\Tilde{T}} \|\Tilde{\ell}(t)\|_p^{p + \log \Tilde{T}} + (5\mu\cdot(p+\log \Tilde{T}))^{p+\log \Tilde{T}}\cdot \mu \cdot \sum_{t=1}^{\Tilde{T}} \|\ell^{*}(t)\|_p^{p + \log \Tilde{T}}\nonumber
\end{align}
}

Inequality \eqref{ineq21} follows by Minkowski's Inequality \cite{B05}.
Inequality \eqref{ineq22} follows as for every $x, y\geq 0, a>0$, and $w\geq 0$ it holds that $(x + y)^{w} \leq (1+a)^w\cdot x^{w} + \left(1+\frac{1}{a}\right)^{w}\cdot y^{w}$ \footnote{\label{comb-lemma}If $y \leq ax$ then, $(x + y)^{w} \leq (x+ax)^w = (1+a)^w\cdot x^w$. Otherwise, $y \geq ax$. In this case we have $(x + y)^{w} \leq (\frac{y}{a} +y)^w = (1+\frac{1}{a})^w \cdot y^w$. Since both terms are positive we conclude the desired.}, when $x=\|\Tilde{\ell}(t)\|_p$, $y=\|\ell^{*}(t)\|_p$, $w=p+\log \Tilde{T}$ and $a=\frac{1}{4\mu w}$ for each $1\leq t \leq \Tilde{T}$ separately.
Inequality \eqref{ineq23} follows since $(1+\frac{1}{4\mu w})^w\leq e^{\frac{1}{4\mu}}\leq 1+\frac{1}{2\mu}$ for $\mu \geq 1$ and $w \geq 1$, where the last inequality holds as  $e^x \leq 1 + 2x$, for  $x \in [0, \ln(2)]$.
Thus,
\begin{align*}
    & \frac{1}{2} \sum_{t=1}^{\Tilde{T}} \|\Tilde{\ell}(t)\|_p^{p + \log \Tilde{T}} 
    \leq (5\mu\cdot(p+\log \Tilde{T}))^{p+\log \Tilde{T}}\cdot \mu  \sum_{t=1}^{\Tilde{T}} \|\ell^{*}(t)\|_p^{p + \log \Tilde{T}}\\
    &\leq (5\mu\cdot(p+\log \Tilde{T}))^{p+\log \Tilde{T}}\cdot \mu \cdot \Tilde{T} \cdot (\max_{t=1}^{\Tilde{T}} \|\ell^*(t)\|_p)^{p+\log \Tilde{T}}
\end{align*}
Putting everything together, we have
\begin{align}
    & (\max_{t}\|\Tilde{\ell}(t)\|_p)^{p+\log \Tilde{T}}
    = \max_{t} \|\Tilde{\ell}(t)\|_p^{p+\log \Tilde{T}}
    \leq \sum_{t=1}^{\Tilde{T}} \|\Tilde{\ell}(t)\|_p^{p+\log \Tilde{T}}  \nonumber
    \\ &\leq 2 \cdot (5\mu\cdot(p+\log \Tilde{T}))^{p+\log \Tilde{T}}\cdot \mu \cdot \Tilde{T} \cdot (\max_{t=1}^{\Tilde{T}} \|\ell^*(t)\|_p)^{p+\log \Tilde{T}} \nonumber
\end{align}
Recall that $\Tilde{\ell}(t)$ is overestimate $\ell(t)$ and $\tilde{T} = 2 \mu_1 \Tilde{D}$.
By taking a power of $\frac{1}{p+ \log \Tilde{T}}$ and noting that $x^{\frac{1}{\log x}} = 2$, we have:
\begin{align}
    & \max_{t=1}^{\tilde{T}} \|\ell(t)\|_p 
    \leq \max_{t=1}^{\tilde{T}} \|\Tilde{\ell}(t)\|_p \nonumber \\
    & 
    \leq 2^{\frac{1}{p+\log \Tilde{T}}} \cdot (5\mu\cdot(p+\log \Tilde{T})) \cdot \mu^{\frac{1}{p+\log \Tilde{T}}} \cdot \Tilde{T}^{\frac{1}{p+\log \Tilde{T}}} \cdot \max_{t=1}^{\tilde{T}} \|\ell^*(t)\|_p \nonumber \\
    &\leq 2 \cdot (5\mu\cdot(p+\log \Tilde{T})) \cdot \mu_1^{\frac{1}{p+\log \Tilde{T}}} \cdot \mu_2^{\frac{1}{p+\log \Tilde{T}}} \cdot 2 \cdot \max_{t=1}^{\tilde{T}} \|\ell^*(t)\|_p \nonumber \\
    &\leq 20 \cdot \mu\cdot(p+\log \Tilde{T}) \cdot \mu_1^{\frac{1}{\log (2\Tilde{\mu_1})}} \cdot \mu_2^{\frac{1}{\log (2\Tilde{D})}} \cdot \max_{t=1}^{\tilde{T}} \|\ell^*(t)\|_p \nonumber \\
    &\leq C \cdot \mu (p+ \log (\mu \tilde{D})) \cdot \max_{t=1}^{\tilde{T}} \|\ell^*(t)\|_p \label{ineq111}\\
    & = \bigO\left(\mu (p+ \log (\mu \tilde{D}))\right) \cdot \max_{t=1}^{\tilde{T}} \|\ell^*(t)\|_p \nonumber
\end{align}
where the last inequality holds since $\mu_2 \leq \tilde{D}$, and for $C \leq 80$.

\end{proof}

\subsection{Removing the Assumptions}\label{remove-assumptions}

The Load Balancing Algorithm in Subsection \ref{lp_norm_algorithm} is based on two assumptions. In this Subsection, we will demonstrate how to eliminate each of these assumptions while maintaining the desired competitive ratio.
To do that, we partition the sequence $\sigma$ into sub-sequences $\sigma_i$.  
For any sub-sequence $\sigma_i$ of $\sigma$ and any online algorithm $ON$ applied to $\sigma_i$, we define $ON(\sigma_i, t)$ as the $\ell_p$-norm of the load vector $\ell^{\sigma_i}(t)$, which represents the loads at time $t$ considering only the jobs derived from $\sigma_i$, independent of other jobs in $\sigma$. Formally,  
$
ON(\sigma_i, t) = \| \ell^{\sigma_i}(t) \|_p
$.
Additionally, denote $ON(\sigma_i)$ as $\max_{t=1}^{T} ON(\sigma_i, t)$. Similar analogies apply to the optimal offline algorithm $OPT$.

\paragraph*{Eliminate assumption \ref{series_distortion_params}}\label{remove:assumption_3}
Our objective is to remove the assumption that the parameters, $\mu_{1}$ and $\tilde{D}$, are known to the algorithm in advance. Note that $\mu$ and $\mu_2$ are only used in the analysis and need not be known by the algorithm.

Our guarantee is that whenever our algorithm is given values $\mu^{'}_1\ge \mu_1$, $\mu^{'} \ge \mu$ and $D^{'} \ge \tilde{D}$
 then, it outputs a solution that is $\bigO\left(\mu^{'} (p+ \log (\mu^{'} \cdot D^{'}))\right)$-competitive. We apply standard (careful) doubling technique to guess these parameters.
 We modify the algorithm and create different copies of the algorithm which run independently on different jobs. The modification is given \hyperref[modification_target]{here}.
\begin{algorithm}[t]
\captionsetup{type=algorithm}
    \phantomsection
    \label{modification_target} 
    \caption*{Modification of the algorithm to remove assumption \ref{series_distortion_params}}
    \begin{algorithmic}[1]
        \Statex 1: $\mu_1^{1} \gets 1 $, {$\tilde{D}_{1} \gets 2^p$}, $i \gets 1$
        \Statex 2: Create a copy of the Load Balancing Algorithm (\ref{lp_norm_algorithm}) with the parameters $\mu_1^{1}$ and {$\tilde{D}_{1}$}.
        \Statex 3: \textbf{Upon a job arrival:}
            \Statex 4: \hspace{\algorithmicindent} Update the real parameters $\mu_1$ and $\tilde{D}$ according to {\bf current} and {\bf previous} jobs
            \Statex 5: \hspace{\algorithmicindent} \textbf{if} {$\mu_1 > \mu_1^{i}$ or {$\tilde{D} > \tilde{D}_{i}$}} \textbf{then}
            \Statex 6: \hspace{\algorithmicindent} \hspace{\algorithmicindent}  $\mu_1^{i+1} \gets 2 \mu_1$, $\tilde{D}_{i+1} \gets \tilde{D}^{2} \cdot \mu_1^{i+1}$
            \Statex 7: \hspace{\algorithmicindent} \hspace{\algorithmicindent} Create a new copy of the Load Balancing Algorithm (\ref{lp_norm_algorithm}) with the parameters $\mu_1^{i+1}$ and {$\tilde{D}_{i+1}$}

            \Statex 8: \hspace{\algorithmicindent} \hspace{\algorithmicindent} $i \gets i + 1$
            
            \Statex 9:  \hspace{\algorithmicindent} \text{Use the latest copy for the current job}
    \end{algorithmic}
\end{algorithm}

The modification of $\mu_1$
and $\tilde{D}$ in line 4 is done according to the the jobs until the $j$'th job. If $\Tilde{d}_j > \Tilde{D}$ then we update $\Tilde{D} = \Tilde{d}_j$.
In order to update $\mu_1$ in line 4 in the Algorithm above, we divide jobs $1,\ldots,j-1$ into two types of jobs.
The first type of jobs are those that have already departed, i.e., their actual durations are known. By considering these jobs, the algorithm can accurately estimate the distortion caused by its predictions. By comparing the predicted durations to the actual durations of these departed jobs, the algorithm updates $\mu_1$. 
The second type of jobs are the active jobs that have not yet departed. For these jobs, the algorithm is unable to determine their exact durations at the time of their arrival. However, it can update the underestimation factor $\mu_1$ based on the progress of these active jobs.
In the modification when we create a new copy of the algorithm, the new copy does not take into account all the previous jobs and the loads are 0 on all machines at that moment.

We now prove the competitive ratio remains at most $\bigO\left(\mu (p+ \log (\mu \tilde{D}))\right)$ even after the modification.
\begin{proof}
    Let $t$ be some time, and suppose until time $t$, $k$ copies were opened. Let $\sigma_i$ be the block of jobs which executed by the $i$'th copy of the algorithm, i.e. $\sigma = \sigma_1 \ldots \sigma_k$.
    We have:
    \begin{align*}
        ON(\sigma, t) 
        = \| \ell(t) \|_p = \| \sum_{i=1}^{k} \ell^{\sigma_i}(t) \|_p
        \leq \sum_{i=1}^{k} \| \ell^{\sigma_i}(t) \|_p 
        = \sum_{i=1}^{k} ON(\sigma_i, t) \leq \sum_{i=1}^{k} ON(\sigma_i)
    \end{align*}
    Let $C$ be the constant from the upper bound proof in inequality \eqref{ineq111}. Let $\mu^{i}= \mu_1^i \cdot \mu_2$ where $\mu_2$ is the real value that may only increase over time. For each copy independently:
    \begin{align*}
        ON(\sigma_i) 
        \leq C \cdot \mu^{i} \cdot (p + \log (\mu^{i} \cdot\tilde{D}_{i})) \cdot OPT(\sigma_i)
        &\leq C \cdot \mu^{i} \cdot (p + \log (\mu^{i} \cdot\tilde{D}_{i})) \cdot OPT(\sigma)\\
        &\leq 2C \cdot \mu^{i} \cdot \log (\mu^{i} \cdot\tilde{D}_{i}) \cdot OPT(\sigma),
    \end{align*}
    where the last inequality holds since the initial value $\tilde{D}_{1} \gets 2^p$ 
and this value only grows over time. Therefore, $\log (\mu^{i} \cdot\tilde{D}_{i}) \geq \log (\tilde{D}_{i}) \geq p$ for each $i$.
    We have:
    \begin{align*}
        ON(\sigma, t) \leq \sum_{i=1}^{k} ON(\sigma_i)
        &\leq  \sum_{i=1}^{k} 2C \cdot \mu^{i} \cdot \log (\mu^{i} \cdot\tilde{D}_{i}) \cdot OPT(\sigma)
        = 2C \cdot \sum_{i=1}^{k} \mu^{i} \cdot \log (\mu^{i} \cdot\tilde{D}_{i}) \cdot OPT(\sigma)
    \end{align*}
    Let $c_i = \mu^{i} \cdot \log (\mu^{i} \cdot\tilde{D}_{i})$.
    We now prove that $ \sum_{i=1}^{k} c_i = \bigO (\mu (p + \log (\mu \tilde{D})))$.
    
    We observe the $i+1$'th copy.
    By the definition of the \hyperref[modification_target]{modification}, the $i+1$'th copy was created by enter the if-statement on line 5.
    if $\mu_1 > \mu_1^{i}$, then  $\mu^{i+1} = \mu_1^{i+1} \mu_2 \geq 2\mu_1^{i} \mu_2 \geq 2 \mu^{i}$
      
 since $\mu_2$ is monotonic non-decreasing. Hence $c_{i+1} \geq 2 c_i$.
    Otherwise, we have $\tilde{D} > \tilde{D}_{i}$, hence:
    \begin{align*}
        c_{i+1} 
        \geq \mu^{i} \cdot \log (\mu^{i} \cdot\tilde{D}_{i+1})
        \geq \mu^{i} \cdot \log (\mu^{i} \cdot \tilde{D}_{i}^{2}\cdot\mu^{i}) = 2c_i
    \end{align*}
    In total we have $c_{i+1} \geq 2c_i$.
    By our update rule we have $\mu_1^{i} \leq 2 \mu_1$, 
    and $ \tilde{D}_{i} \leq  \max\{ 2^{p},\tilde{D}^{2} \cdot 2\mu_1\}
    = \max\{2^p, \tilde{D}^{2} \cdot 4\mu\}$, for each $i$, for the final real parameters $\mu_1$
    and $\tilde{D}$ of the series.
    Therefore, we conclude:
    \begin{align*}
        \sum_{i=1}^{k} c_i 
        \leq 2 \cdot c_k 
        = 2 \cdot \mu^{k} \cdot \log (\mu^{k} \cdot\tilde{D}_{k}) 
        &\leq 2 \cdot 4\mu \cdot \log(\max\{ 2^p, \tilde{D}^{2} \cdot 4\mu\} \cdot 4\mu) = \bigO (\mu (p + \log (\mu \tilde{D})))
    \end{align*}
    Therefore, for any $t$, we have $ON(\sigma, t) \leq \bigO (\mu (p + \log (\mu \tilde{D}))) \cdot OPT(\sigma)$ which concludes the proof.
\end{proof}

\paragraph*{Eliminate assumption \ref{series_end_time}}\label{series_end_time_removal}
Our objective is to remove the assumption that all jobs arrive in the time range 0 up to $\mu_1 \tilde{D}$.

To accomplish this, we assume, without loss of generality, that the arrival time of the first job is $t_{1} = 0$ and we divide the jobs into groups based on their arrival times. Let $\Tilde{T} \triangleq 2\mu_1 \tilde{D}$. Each group has an interval length of $\frac{\Tilde{T}}{2}$, and group $k$ consists of all jobs that arrive within the time interval $[(k-1)\frac{\Tilde{T}}{2}, k\frac{\Tilde{T}}{2})$. We denote group $k$ by $\sigma_{k}$ and $\sigma = \sigma_{1},\ldots,\sigma_{k}$. It is established that for job $j$ departs by time $t_{j} + d_{j} \leq t_{j} + \mu_{1} \tilde{D} 
= t_{j} + \frac{\Tilde{T}}{2}$. We conclude that jobs which belong to group $k$ must depart by time $(k+1)\frac{\Tilde{T}}{2}$.

To address the assumption, we employ a distinct instance of the Load Balancing Algorithm (\ref{lp_norm_algorithm}) for each group.
This means that jobs within a specific group are assigned independently of the task assignments in other groups, as shown in figure \ref{fig:assumption_1_removal}.
It is important to note that for each group $k$ it holds that,
\begin{align*}
    ON(\sigma_{k}) \leq C \cdot (\mu (p+ \log \Tilde{T}) \cdot OPT(\sigma_{k}) \leq C \cdot (\mu (p+ \log \Tilde{T}) \cdot OPT(\sigma)
\end{align*} Where we use the constant $C$ from the upper bound proof in inequality $\eqref{ineq111}$, and the last inequality follows since $\sigma_{k}$ is a sub-series of $\sigma$.

We now prove that the competitive-ratio is preserved. Let $t$ be some time where there are active jobs and observe time $t$. Let $k$ be such that $t \in [(k-1)\frac{\Tilde{T}}{2}, k\frac{\Tilde{T}}{2})$. $WLOG$ we assume that $k \geq 2$ (the case $k=1$ is trivial). As previously stated, the active jobs consist solely of the jobs belonging to groups $k-1$ or $k$.
Based on the aforementioned observations, we conclude the following:
\begin{align}
    ON(\sigma,t) 
    &= \| \ell^{\sigma_{k-1}}(t) + \ell^{\sigma_k}(t) \|_p
    \leq \| \ell^{\sigma_{k-1}}(t) \|_p + \| \ell^{\sigma_k}(t) \|_p \label{ineq31}
    \\ &= ON(\sigma_{k-1}, t) + ON(\sigma_{k}, t)
    \leq ON(\sigma_{k-1}) + ON(\sigma_{k})\nonumber
    \\ &\leq
    C \cdot (\mu (p+ \log \Tilde{T}) \cdot OPT(\sigma) +
    C \cdot (\mu (p+ \log \Tilde{T})\cdot OPT(\sigma) 
    \\ &= \bigO\left(\mu (p+ \log (\mu \tilde{D}))\right)\cdot OPT(\sigma) \nonumber
\end{align}
Inequality \eqref{ineq31} follows by Minkowski's Inequality \cite{B05}.

\noindent Since it holds for every $t$, we conculde that $ON(\sigma) = \max_{t=1}^{T} ON(\sigma, t) \leq \bigO\left(\mu (p+ \log (\mu \tilde{D}))\right)\cdot OPT(\sigma)$, Hence the assumption can be removed.

\begin{figure}[t]
    \centering
    \includegraphics[width=15cm]{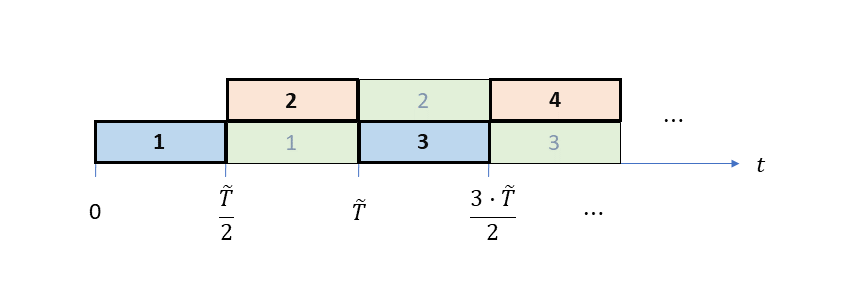}
    \caption{The blue and orange boxes represent distinct time intervals, with a separate instance of the Load Balancing Algorithm (\ref{lp_norm_algorithm}) executed within each interval. In the context of these intervals, a green box labeled with the number "$k$" indicates that jobs from the $k$-th copy of the algorithm can still be active throughout that specific green interval. These jobs terminate before the green interval concludes.}
    \label{fig:assumption_1_removal}
\end{figure}

\section{Lower Bounds}\label{sec:lower}

In this Section we prove Theorem \ref{thm:lower} that we repeat here for convenience.

\begin{theorem}

    Let $p\in [1,\log m]$, $\mu\leq m$, $\tilde{D}\leq {m}^{m}$. Then, any deterministic online algorithm for the load balancing of temporary tasks with has a competitive ratio of at least \\
    $\Omega\left(\mu^{\frac{1}{2} - \theta(\frac{1}{p})} + p+(\frac{\log \tilde{D}}{\log \log \tilde{D}})^{\frac{1}{2}-\theta(\frac{1}{p})}\right)$, where $\frac{1}{2}-\theta(\frac{1}{p})=\frac{p}{2p-1}-\frac{1}{p}$.
    In particular, for $\ell_{\infty}$, the competitive ratio is at least
    $\Omega\left(\sqrt{\mu}+\log m+\sqrt{\frac{\log \tilde{D}}{\log \log \tilde{D}}} \right)$. This lower bounds also hold for the restricted assignment model.
\end{theorem}

We prove the lower bound in three steps: 
$\Omega(p)$,
$\Omega\left(\mu^{\frac{1}{2}-\theta(\frac{1}{p})}\right)$ and 
$\Omega\left((\frac{\log \tilde{D}}{\log \log \tilde{D}})^{\frac{1}{2}-\theta(\frac{1}{p})}\right)$.
We note that the $\Omega(p)$ lower bound is proved for $\mu = 1$ and $\tilde{D} = 1$;
the lower bound of $\Omega\left(\mu^{\frac{1}{2}-\theta(\frac{1}{p})}\right)$ is proved for $\tilde{D}=1$ and $p \leq \mu^{\frac{1}{2}-\theta(\frac{1}{p})}$;
and the $\Omega\left((\frac{\log \tilde{D}}{\log \log \tilde{D}})^{\frac{1}{2}-\theta(\frac{1}{p})}\right)$ lower bound is proved for $\mu = 1$ and $\log m \leq (\frac{\log \tilde{D}}{\log \log \tilde{D}})^{\frac{1}{2}-\theta(\frac{1}{p})}$.
Hence, the lower bound holds for any admissible parameters.

First we prove the $\Omega(p)$ lower bound. This is done by simulating the lower bound of $\Omega(p)$ for the permanent jobs case (as mentioned in \cite{AAGKKV95}), where the predictions of the durations are set to be the exact real durations (thus eliminating the need for predictions). We conclude that the competitive ratio of any deterministic algorithm for the online $\ell_p$-norm minimization of temporary tasks with predictions is at least $\Omega(p)$.
Furthermore, for the online $\ell_\infty$-norm minimization of temporary tasks with predictions, the competitive ratio is proven to be at least $\Omega(\log m)$.

Next, we show a lower bound as a function of $\mu$.
\begin{lemma} \label{sqrt_mu_lower_bound}
    The competitive ratio of any deterministic algorithm for the online $\ell_p$-norm minimization of temporary tasks with predictions with distortion $\mu \leq m$ is at least $\Omega\left(\mu^{\frac{p}{2p-1}-\frac{1}{p}}\right) = \Omega\left(\mu^{\frac{1}{2}-\theta(\frac{1}{p})}\right)$.
\end{lemma}

\begin{remark}
    Specifically, in the $\ell_{\infty}$-norm, the competitive ratio of any algorithm with distortion $\mu \leq m$ is at least $\Omega\left(\sqrt{\mu}\right)$.
\end{remark}

\begin{proof}[Proof of Lemma \ref{sqrt_mu_lower_bound}]
    Let $ON$ be an online algorithm for the problem, and let $\mu \geq 1$.
    For convenience we let $R \triangleq \min\{\mu, \frac{m}{2}\}$, and define $k \triangleq R^{\frac{p}{2p-1}}$.
    We will construct a job sequence $\sigma$ with a distortion of $R$, such that the competitive ratio of $ON$ on this specific input is at least $\Omega\left(R^{\frac{p}{2p-1}-\frac{1}{p}}\right)$.

    \paragraph*{Sequence definition:}

    The sequence consists of a maximum of $R$ iterations. Each iteration is associated with a unique machine, and the jobs within each iteration can only be assigned to either the first $k$ machines or the machine specific to that iteration. The sequence is structured in the following manner: during iteration $r$, all the jobs associated with that iteration are released by time $t_r$, and their predicted duration is $\tilde{d_j} = 1$. If the online algorithm assigns a job to the unique machine, the job departs right before the beginning of the next iteration. On the other hand, if the online algorithm $ON$ assigns a job to one of the first $k$ machines, we denote that specific job as $j_r$. In this case, the job departs from the machine at time $R$ and we start the next iteration.

    \begin{algorithm}[t]
    \phantomsection
    \label{lower_bound_1_description}
    \captionsetup{type=algorithm}
    \caption*{Description of the lower bound of Lemma \ref{sqrt_mu_lower_bound}}
    \begin{algorithmic}
        \State $t_1 \gets 0$
        \For{$r \gets 1$ to $R$}
    
            \For{$j \gets 1$ to $k$}
                \State Release job $j$ by time $t_j = t_r$ with $\tilde{d_j}=1$ and
                \[p_{ij}= \left\{\begin{array}{ll}
                1 & i\in \{1, \dots, k, k+r\}\\
                \infty & \mbox{Otherwise}\end{array}\right.\]
                
                \If{$ON$ assigns the job to machine $k+r$}
                    \State $d_j \gets 1$
                \ElsIf{$ON$ assigns the job to machine $\in \{1, \dots, k \}$} 
                    \State $d_j \gets R - t_j$
                    \State break   \hfill\COMMENT{//start the next iteration}
                \EndIf
                
            \EndFor
            
            \If{$ON$ has load of $k$ on some machine}
                \State break \hfill\COMMENT{//start the next iteration}
            \EndIf
            \State $t_{r+1} \gets t_r + 1$
            
        \EndFor
        
    \end{algorithmic}
    \end{algorithm}

    It is worth noting that the sequence is well-defined since $k+R \leq m$, and the distortion of the sequence is at most $R$.
    To provide a comprehensive understanding of the sequence, a detailed \hyperref[lower_bound_1_description]{description} is given.
    
    \paragraph*{Analysis:}
    It is evident that in each iteration $r$, $OPT$ would assign job $j_r$ to machine $k+r$, while the remaining $k-1$ jobs would be evenly distributed, with one job per machine, among a subset of $k-1$ machines chosen arbitrarily. Since we have a maximum of $R$ iterations, it implies that at any given time, $OPT$ has at most $k+R$ active jobs distributed among $k+R$ machines.
    Therefore, 
    \begin{align*}
        OPT(\sigma)
        = \max_{t=1}^{R} \|\ell^{*}(t)\|_{p} 
        = \max_{t=1}^{R} (\sum_{i=1}^{m} \ell^{*}_{i}(t)^{p})^{\frac{1}{p}} 
        \leq (k+R)^{\frac{1}{p}}
        \leq (2R)^{\frac{1}{p}}
    \end{align*}   
    On the other hand, if the sequence terminates at some iteration $r < R$, then $ON$ assigns all $k$ jobs to machine $k+r$. As a result, by that time, the load on that particular machine amounts to $k$, and so
    \begin{align*}
    ON(\sigma) =  \max_{t=1}^{{r}} ON(t) \geq ON(r) = \|\ell(r)\|_{p}  \geq (k^{p})^{\frac{1}{p}} = k
    \end{align*}
    If the sequence terminates after $R$ iterations, then there are $R$ active jobs, each of which is assigned to one of the first $k$ machines. Hence,
    \begin{align}
        ON(\sigma) 
        &= \max_{t=1}^{R} ON(t) 
        \geq ON(R)
        = (\sum_{i=1}^{m} \ell_{i}(R)^{p})^{\frac{1}{p}}
        \geq (\sum_{i=1}^{k} \ell_{i}(R)^{p})^{\frac{1}{p}} \nonumber
        \\& \geq (k(\frac{R}{k})^{p})^{\frac{1}{p}} \label{ineq61}
        \\ &= (k^{1-p} R^{p})^{\frac{1}{p}} 
        = (k^{1-p} (k^{\frac{2p-1}{p}})^{p})^{\frac{1}{p}} = k \nonumber
    \end{align}
    Inequality \eqref{ineq61} holds by using Jensen's Inequality with the convex function $f(x) = x^{p}$, it holds:
    $ f(\frac{1}{k}\sum_{i=1}^{k} \ell_{i}(R))\leq \frac{1}{k}\sum_{i=1}^{k} f(\ell_{i}(R))$. By applying the constraint $\ell_{1}(R) + \ldots + \ell_{k}(R) = R$, we finish.
    
    Therefore, the competitive ratio is lower bounded by:
    \begin{align*}
        \frac{ON(\sigma)}{OPT(\sigma)}
        \geq \frac{k}{(2R)^{\frac{1}{p}}}
        = \frac{R^{\frac{p}{2p-1}}}{(2R)^{\frac{1}{p}}}
        = \Omega\left(R^{\frac{p}{2p-1}-\frac{1}{p}}\right)
    \end{align*}
    \noindent Note that for the $\ell_{\infty}$-norm, as $p\to \infty$ this yields us a lower bound of $\Omega\left(\sqrt{R}\right)$.
\end{proof}
Finally, we show a lower bound as a function of $\tilde{D}$.
\begin{lemma} \label{sqrt_logD_lower_bound}
    The competitive ratio of any deterministic algorithm for the online $\ell_p$-norm minimization of temporary tasks with predictions with $\tilde{D} \leq m^m$ is at least $\Omega\left((\frac{\log \tilde{D}}{\log \log \tilde{D}})^{\frac{p}{2p-1}-\frac{1}{p}}\right) = \Omega\left((\frac{\log \tilde{D}}{\log \log \tilde{D}})^{\frac{1}{2}-\theta(\frac{1}{p})}\right)$.
\end{lemma}

\begin{remark}
    Specifically, in the $\ell_{\infty}$-norm, the competitive ratio of any algorithm with distortion $\tilde{D} \leq m^m$ is at least $\Omega\left(\sqrt{\frac{\log \tilde{D}}{\log \log \tilde{D}}}\right)$.
\end{remark}

\begin{proof}[Proof of Lemma \ref{sqrt_logD_lower_bound}]
    The proof is similar to the the proof of Lemma \ref{sqrt_mu_lower_bound}.

    \begin{algorithm}[t!]
    \phantomsection
    \label{lower_bound_2_description}
    \captionsetup{type=algorithm}
    \caption*{Description of the lower bound of Lemma \ref{sqrt_logD_lower_bound}}
    \begin{algorithmic}
        \State $t_1 \gets 0$
        \State $x_1 \gets \tilde{D}$
        \For{$r \gets 1$ to $R$}
    
            \For{$j \gets 1$ to $k$}
                \State Release job $j$ by time $t_j = t_r$ with departure time $t_j + \tilde{d_j} = \frac{x_r}{2}(1+\frac{j}{k})$ and
                \[p_{ij}= \left\{\begin{array}{ll}
                1 & i\in \{1, \dots, k, k+r\}\\
                \infty & \mbox{Otherwise}\end{array}\right.\]

                \If{$ON$ assigned the job to machine $\in \{1, \dots, k \}$ AND $j=1$} 
                    \State $t_{r+1} \gets \frac{x_r - t_r}{2}$
                    \State $x_{r+1} \gets$ departure's time of job $j=1$
                    \State break  \hfill\COMMENT{//start the next iteration}
                \EndIf
                
                \If{$ON$ assigned the job to machine $\in \{1, \dots, k \}$ OR $j=k$} 
                    \State $t_{r+1} \gets$ departure's time of job $j-1$ of the current iteration.
                    \State $x_{r+1} \gets$ departure's time of job $j$ of the current iteration.
                    \State break \hfill\COMMENT{//start the next iteration}
                \EndIf
                
            \EndFor
            \If{$ON$ has load of $k$ on some machine}
                \State break \hfill\COMMENT{//start the next iteration}
            \EndIf
            
        \EndFor
        
    \end{algorithmic}
    \end{algorithm}
    
    Let $ON$ be an online algorithm for the problem, and let $\tilde{D} \geq 1$.
    We create a jobs sequence with no estimations, i.e. $\mu=1$, hence in this case $D = \tilde{D}$.
    For convenience we let $R \triangleq \min\{\lfloor x \rfloor, \frac{m}{2}\}$, for some $x$ which solves the equality $\tilde{D} = \sqrt{x}^{x-\sqrt{x} + 1}$, and define $k \triangleq R^{\frac{p}{2p-1}}$.
    We will construct a job sequence $\sigma$ with $\tilde{D}(\sigma) = \tilde{D}$, such that the competitive ratio of $ON$ on this specific input is at least $\Omega\left(R^{\frac{p}{2p-1}-\frac{1}{p}}\right)$.

    \paragraph*{Sequence definition:}
    
    The sequence consists of a maximum of $R$ iterations. Each iteration is associated with a unique machine, and the jobs within each iteration can only be assigned to either the first $k$ machines or the machine specific to that iteration. The sequence is structured in the following manner: during iteration $r$, the departure times of the jobs are arranged in ascending order, meaning that the departure time of each job is greater than the departure time of the preceding job within the same iteration. As the next iteration, $r+1$, begins, all the jobs from the previous iteration, $r$, depart, except for the last job of iteration $r$, denoted as job $j_r$. Job $j_r$ remains active until the completion of the entire series.

    To provide a comprehensive understanding of the sequence, a detailed \hyperref[lower_bound_2_description]{description} along with a \hyperref[fig:lower_bound_D]{figure} are given.

    \begin{figure}[H]
    \centering
    \includegraphics[width=12cm]{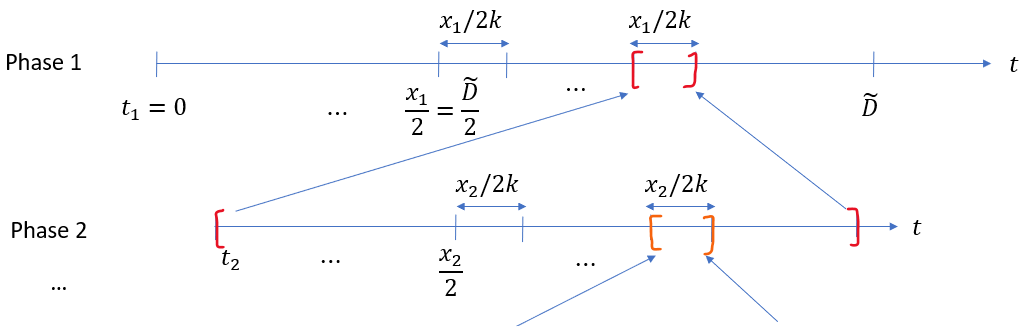}
    \caption{Phase 1 demonstrates the release of all jobs by time $0$, sorted in ascending order based on their durations. The gap between their departure times is consistent at $x_1/2k$. Once the algorithm assigns a job to any of the initial $k$ machines, we replicate the phase 1 steps for the time interval between the departure times of the last two jobs in phase 1. Subsequently, in Phase $i$, we maintain a consistent gap of $x_i/2k$ between the durations of the jobs.}
    \label{fig:lower_bound_D}
    \end{figure}
    
    \noindent The part of the analysis is the same as in Lemma \ref{sqrt_mu_lower_bound}.
    By noting that $ x = \Omega\left(\frac{\log \tilde{D}}{\log \log \tilde{D}}\right)$, we get the desired.
\end{proof}

\section{The Price of Estimation}\label{sec:price}

In this Section we prove Theorem \ref{thm:PoE}.

\noindent Note that within the results of the theorem we use $\tilde{D}$ although the concept of {\em Price of Estimation} does not contain estimations and it is actually an offline problem. However, in the case of no estimations, we have $D = \tilde{D}$. In our proof, we will use the notation $D$ instead of $\tilde{D}$.
In addition, by Observation \ref{POE=POE_1}, it is enough to show that $PoE_1(\mu) = \Theta(\mu \log D)$.

To establish the theorem, encompassing both the lower and upper bounds, we introduce a sequence of time points $t_1, \ldots, t_{j*}$, determined as follows:
\begin{algorithm}
\captionsetup{type=algorithm}
\caption*{Time points determination given $D, \mu$}
\begin{algorithmic}\label{times-series-def}
    \State $t_1 \gets 0$
    \State $t_2 \gets t_{1} + \frac{D-t_1}{\mu}$
    \State $j \gets 2$
    \While{$t_j - t_{j-1} \geq 1$}
        \State $t_{j+1} \gets t_{j} + \frac{D-t_j}{\mu}$ \label{t_j+1 def}
        \State $j \gets j + 1$
    \EndWhile
    \State $j^{*} \gets j - 1$
\end{algorithmic}
\end{algorithm}

\begin{figure}[htbp]
\centering
\includegraphics[width=12cm]{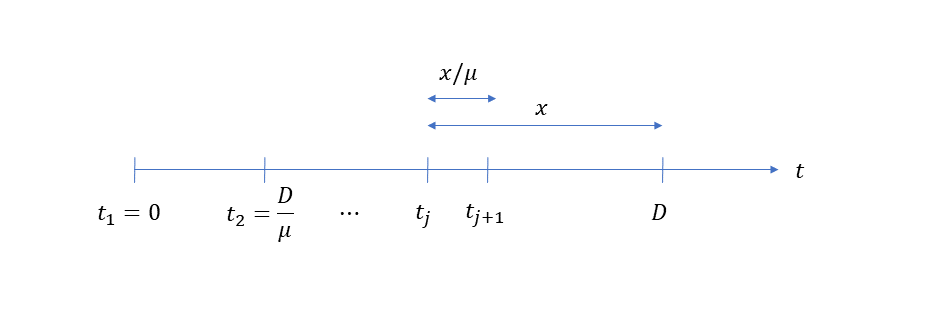}
\caption{Definition of $t_{j+1}$}
\label{fig:lower_bound_opt_comp}
\end{figure}

\begin{lemma}\label{induction_lemma}
        It holds that for every $j \geq 2$, $t_{j} = D \cdot \left[1-(1-\frac{1}{\mu})^{j-1}\right]$.
\end{lemma}
\begin{proof}
    We will prove the Lemma using induction on $j$.

    \noindent \textbf{Base case (j = 2):  } By definition, we have $t_{1} = \frac{D}{\mu} = D \cdot \left[ 1- (1-\frac{1}{\mu})^{2-1} \right]$.\\
    \textbf{Inductive Step: } 
    Assuming that the Lemma holds for index $j$, let's proceed with the proof for index $j+1$. We have the following:
    \begin{align}
        t_{j+1} 
        &= t_j + \frac{D-t_j}{\mu} 
        = (1-\frac{1}{\mu})\cdot t_j + \frac{D}{\mu} \label{ineq71}\\
        &= (1-\frac{1}{\mu})\cdot D \cdot \left[1-(1-\frac{1}{\mu})^{j-1}\right] + \frac{D}{\mu} \label{ineq72}\\
        &= D \cdot \left[(1-\frac{1}{\mu}) - (1-\frac{1}{\mu})^{j}\right] + \frac{D}{\mu} = D \cdot \left[1-(1-\frac{1}{\mu})^{j}\right]\nonumber
    \end{align}
    Equality \eqref{ineq71} follows by the definition of $t_{j+1}$, and equality \eqref{ineq72} follows by the inductive step.
\end{proof}

\begin{lemma}\label{opt_lower_bound_lemma}
    $j^{*} = \Theta(\mu \log D)$.
\end{lemma}
\begin{proof}
    Based on the definition of $\sigma$, $j^*$ is determined as the first index $j$ for which $t_{j+1}-t_j \leq 1$.
    We utilize Lemma \ref{induction_lemma} and deduce the following:

    \begin{align*}
        t_{j^{*}+1} - t_{j^{*}} 
        &= D \cdot \left[1-(1-\frac{1}{\mu})^{j^{*}}\right] - D \cdot \left[1-(1-\frac{1}{\mu})^{j^{*}-1}\right] \\
        &= D \cdot (1-\frac{1}{\mu})^{j^{*}-1} - D \cdot (1-\frac{1}{\mu})^{j^{*}}
        = D \cdot (1 - \frac{1}{\mu})^{j^{*}-1} \left[1- (1-\frac{1}{\mu})\right] \\
        &= \frac{D}{\mu} \cdot (1-\frac{1}{\mu})^{j^{*}-1}
        \leq 1
    \end{align*}

    \noindent Hence we have $\frac{D}{\mu} \leq (\frac{\mu}{\mu -1})^{j^{*}-1} = (1 + \frac{1}{\mu - 1})^{j^{*}-1}$, and by taking a $\log_{1 + \frac{1}{\mu - 1}}$ we achieve,
    \begin{align*}
        j^{*} 
        = \Omega\left( \log_{1+\frac{1}{\mu -1}} \frac{D}{\mu}\right)
        = \Omega\left( \frac{\log \frac{D}{\mu}}{\log (1+\frac{1}{\mu -1 })}\right)
        = \Omega\left(\mu \log \frac{D}{\mu} \right) 
        = \Omega\left(\mu \log D \right)
    \end{align*}

    \noindent In addition, with similar calculations and by observing the fact that $t_{j^{*}} - t_{j^{*}-1} \geq 1$, we conclude that $\frac{D}{\mu} \geq (\frac{\mu}{\mu -1})^{j^{*}-2} = (1 + \frac{1}{\mu - 1})^{j^{*}-2}$, and by taking a $\log_{1 + \frac{1}{\mu - 1}}$ we achieve,
    \begin{align*}
        j^{*} 
        = \bigO\left( \log_{1+\frac{1}{\mu -1}} \frac{D}{\mu}\right)
        = \bigO\left( \frac{\log \frac{D}{\mu}}{\log (1+\frac{1}{\mu -1 })}\right)
        = \bigO\left(\mu \log \frac{D}{\mu} \right) 
        = \bigO\left(\mu \log D \right)
    \end{align*}
\end{proof}

We now prove Theorem \ref{thm:PoE}.
\begin{proof}[Proof of upper bound of Theorem \ref{thm:PoE}]\label{PoE-upper-proof}
Given any instance $I$ of jobs with a single machine, our goal is to show that,
\begin{align*}
    \max_{t=1}^{T} \ell(I(\mu), t) \leq \bigO\left(\mu \log D \right) \cdot \max_{t=1}^{T} \ell(I, t)
\end{align*}
Let $\tilde{t} \in [1, T]$, and let $\tilde{t}_0 \triangleq \max \{ 1, \tilde{t} - D\}$. We observe the \hyperref[times-series-def]{Time points determination} and we employ a similar interval division approach on the interval $[\tilde{t}_0, \tilde{t}]$ as we previously did for the interval $[0, D]$. By doing so, we can create $j^{**}$ intervals, where $j^{**} \leq j^{*}$ due to the fact that $[\tilde{t}_0, \tilde{t}]$'s length is at most $D$.
We break the $i$'th machine's load at time $\tilde{t}$ by instance $I(\mu)$ to $\ell(I(\mu), \tilde{t}) = \ell(I, \tilde{t}_0) + \ell(I, \tilde{t}) + X$, where $X$ accounts for the load of jobs with arrival time in $[\tilde{t}_0, \tilde{t}]$ and which by instance $I$ are already considered as departed at time $\tilde{t}$ while for the jobs instance $I(\mu)$ they are still active.

Next, we decompose $X$ into smaller constituent parts: $X = X_1 + \ldots + X_{j^{**}} + Y$, with $X_k$ signifying the loads of jobs whose arrival times fall within the $k$-th interval. The term $Y$ covers all jobs with arrival times within the interval $[t_{j^{**}}, \tilde{t})$. Notably, the length of the interval $[t_{j^{**}}, \tilde{t})$ does not exceed $\mu +1$, allowing us to break it down into at most $\mu + 1$ intervals, each with a length of 1 (excluding the last interval which may be shorter). In total, we express $X = X_1 + \ldots + X_{j^{**}} + Y_1 + \ldots + Y_{\mu +1}$, where $Y_k$ signifies the load of jobs from the k-th interval ($[t_{k-1}^{''}, t_{k}^{''})$), with $t_{k}^{''} = t_{j^{**}} +k$.

The minimum duration of each job is 1. Thus, all jobs in $I$ with arrival times in the interval $[t_{k-1}^{''}, t_{k}^{''})$ must be active by the time $t_{k}^{''}$ (since that interval length is at most 1), so the load the machine by time $t_{k}^{''}$ is at least the load of these jobs, i.e. for every $k$, $Y_k \leq \ell(I,t_{k}^{''})$.

Furthermore, based on the definition of our sequence $t_{1}^{'}, t_{2}^{'}, \ldots, t_{j^{}}^{'}$, it follows that for every $k \leq j^{} - 1$, if a job arrives in the $k$-th interval with a certain duration and by extending its duration by a factor of $\mu$ it remains active until time $\tilde{t}$, then it must also be active by the time $t_{k+1}^{'}$. Otherwise, it would not remain active beyond time $t_{j+1}^{'}$.
This implies that for each $k \leq j^{**}-1$, we have $X_k \leq \ell(I, t_{k}^{'})$.

\noindent Hence we have that,
\begin{align*}
    \ell(I(\mu), \tilde{t}) 
    &\leq \ell(I, \tilde{t}_0) + \ell(I, \tilde{t})
    + \sum_{k=1}^{j^{**}} \ell(I,t_{k}^{'}) 
    + \sum_{k=1}^{\mu+1} \ell(I,t_{k}^{''}) \\
    &\leq (2+j^{**}+\mu+1) \cdot \max_{t=1}^{T} \ell(I, t) 
    \leq (2+j^{*}+\mu+1) \cdot \max_{t=1}^{T} \ell(I, t)  \\
    &= \bigO\left(\mu \log D \right) \cdot \max_{t=1}^{T} \ell(I, t)
\end{align*}
Since it holds for every $\tilde{t}$, we concluded that,
\begin{align*}
    \max_{t=1}^{T} \ell(I(\mu), t) \leq \bigO\left(\mu \log D \right) \cdot \max_{t=1}^{T} \ell(I, t)
\end{align*}

\end{proof}

\begin{proof}[Proof of lower bound of Theorem \ref{thm:PoE}]\label{PoE::_lower_proof}
    To prove this, we will construct an instance of jobs $I$ with a single machine, such that $\ell(I(\mu), D-\epsilon) = \Omega(\mu \log D)$ while $\ell(I, t) = 1$ for every $t$.

    The jobs instance $I$ is defined such that each job $j$ is alive during the time interval $[t_j, t_{j+1})$, which we refer to as the j-th interval. This means that at any given time, only one job is active. However, in the case of $I(\mu)$, each job will be considered active from its arrival time $t_j$ until the end time $D$. Therefore, by observing time $D-\epsilon$, where $\epsilon$ is a small positive value, all the jobs in the sequence are active. Formal definition of the jobs sequence:
    \begin{algorithm}
    \captionsetup{type=algorithm}
    \caption*{Jobs instance $I$ definition}
    \begin{algorithmic}
        \For{$j \gets 1$ to $j^{*}$}
            \State Release job $j$ by time $t_j$ with $p_{1,j} = 1$ and $d_j = t_{j+1} - t_j$
        \EndFor
    \end{algorithmic}
    \end{algorithm}
    
    According to $I$'s definition, $t_j + d_j = t_{j+1}$, which means that job $j$ departs immediately before the arrival of job $j+1$. Therefore, at any given time $t$, the load on the machine for $I$ is only 1, i.e. $\ell(I, t)=1$.
    On the other hand, $I(\mu)$'s jobs durations are $\mu$ times bigger. Hence, according to $I$'s definition, that means that all the jobs depart by the same time $D$. Therefore, by observing time $D-\epsilon$, where $\epsilon$ is a small positive value, the load on the single machine is equal to $j^*$, which is defined to be the number of jobs in the series. Finally, using Lemma \ref{opt_lower_bound_lemma}, we get the desired $\ell(I(\mu), D-\epsilon) = j^{*} = \Omega(\mu \log D)$.
\end{proof}

\section{Discussion and Open Problems}

An obvious open problem is closing the gap between the upper bounds and the lower bounds. For example, for the $\ell_{\infty}$ objective our algorithm is $\bigO\left(\mu\log (m\mu \tilde{D})\right)$-competitive, while the lower bounds are essentially  $\Omega\left(\sqrt{\mu}\right)$, $\Omega\left(\log m\right)$, and $\Omega\left(\sqrt{\frac{\log \tilde{D}}{\log \log \tilde{D}}}\right)$. Closing the quadratic gap in the logarithmic dependency on $D$ is an open problem even for the known duration case~\cite{AKPPW97,AAE03}.
Our focus in this work is the dependency on the duration distortion $\mu$. For this parameter there is also a quadratic gap between our lower bound of $\sqrt{\mu}$ and the upper bound that is linear in $\mu$.
We believe that it may be possible to improve the upper bound further. However, we observe that given the values $\mu_1$ and $\mu_2$, our algorithm is based on the pseudo loads using the estimated durations and never uses the actual duration. We prove in full version of the paper that any algorithm which is oblivious to the actual durations and only based on the estimations is at least $\Omega(\mu)$-competitive, for the maximum load objective. 
Hence, such an algorithm fails to improve the dependency on $\mu$. We remark that the gap is even larger for small values of $p$. For example, for the $\ell_2$ objective our lower bound is only $\Omega\left(\mu^{\frac{p}{2p-1}-\frac{1}{p}}\right)=\Omega\left(\mu^{1/6}\right)$.

\bibliography{Arxiv/refs}
\bibliographystyle{alpha}

\begin{appendices}
\section{A Counter Example for the Naive Algorithm}\label{counter_example}

By applying the best-known algorithm for load balancing of temporary tasks with known durations in the $\ell_\infty$-norm, which is shown in \cite{AKPPW97}, as black box by the upper bound in Section \ref{sec:price} the competitive ratio will be at most $O(\mu \log \tilde{D})$ bigger, i.e the competitive ratio will be $\bigO\left(\mu \log \tilde{D} (\log m + \log \tilde{D}) \right)$.
Here we demonstrate that the algorithm also for the general $\ell_p$-norm indeed reaches a bound, which is almost as high as claimed above and hence using the original algorithm as black box has worse performance.

\begin{lemma}
    By applying the best-known algorithm for load balancing of temporary tasks with known durations as shown in \cite{AKPPW97}, and considering $\tilde{d}_j$ the actual duration of each job $j$ as, the achievable competitive ratio is only lower bounded by $\Omega\left((\mu \log \tilde{D} \log m)^{1-\frac{1}{p}} \right)$.
\end{lemma}

\begin{proof}
    The counterexample presented here is a generalization of the proof of the lower bound in Section \ref{sec:price}. We construct a jobs sequence $\sigma$ using the same idea with the arrival times $t_j$ as defined in (\hyperref[times-series-def]{Time points determination}). However, in this modified series, instead of only one job arriving at time $t_j$ and departing at time $t_{j+1}$ as described in the proof of the lower bound in Section \ref{sec:price}, there will be $m-1$ jobs arriving at time $t_j$.
    The sequence is defined as follows:
    \begin{algorithm}[H]
    \captionsetup{type=algorithm}
    \caption*{A counter example jobs sequence $\sigma$}
    \begin{algorithmic}
        \For{$j \gets 1$ to $j^{*}$}
            \State By time $t_j$ we release the following $m-1$ jobs, all with $\tilde{d_j} = t_{j+1} - t_j$:

            \For{$k \gets 1$ to $\log m$}
                \For{$z \gets 1$ to $\frac{m}{2^k}$}
                \State Release job $j' = (j^{*}, k, z)$ with $p_{z,j'} = p_{z+\frac{m}{2^k},j'} = 1$ and $p_{i, j'} = \infty$ for all $i \neq z, z+\frac{m}{2^k}$ 
                \State Job $j'$ can be assigned only to machines $z$ and $z+\frac{m}{2^k}$ 
                \State WLOG the algorithm assigns the job j on machine $z$, otherwise we rename these two machines.
                \EndFor
            \EndFor
        \EndFor
    \end{algorithmic}
    \end{algorithm}

    \noindent We first analyze $OPT$'s objective. $OPT$ would assign each job $j' = (j^{*}, k, z)$ on machine $z+\frac{m}{2^k}$, yielding a load of $1$ on $m-1$ machines during each iteration $j$.
    Hence, $OPT(\sigma) = (j^{*} \cdot \log m)^{\frac{1}{p}} = \bigO\left((\mu \log \tilde{D} \log m)^{\frac{1}{p}} \right)$.
    On the other hand, the online algorithm adds by the end of each iteration $j$ a load of $\log m$ on the first machine.
    Hence, in total by time $\tilde{D} - \epsilon$ the load on that machine is $j^{*} \log m$ and hence $ON(\sigma) = j^{*} \log m = \mu \log \tilde{D} \log m$.

    \noindent Therefore we have,
    \begin{align*}
        \frac{ON(\sigma)}{OPT(\sigma)}
        \geq \frac{\mu \log \tilde{D} \log m}{\bigO\left((\mu \log \tilde{D} \log m)^{\frac{1}{p}} \right)}
        \geq \Omega\left((\mu \log \tilde{D} \log m)^{1-\frac{1}{p}} \right)
    \end{align*}
    For the special case, in the $\ell_{\infty}$-norm the algorithm reaches the bound of $\Omega\left(\mu \log \tilde{D} \log m \right)$.
\end{proof}

\section{A Lower Bound for Online Algorithms Using Only the Estimations}\label{pseudo-lower-bound}

\begin{lemma}
    Any deterministic online algorithm which is oblivious to the actual durations and only based on the estimations is at least $\Omega(\mu)$ for the maximum load objective.
\end{lemma}

\begin{proof}
Let $ON$ be an online algorithm which using the pseudo loads of jobs. We construct a jobs sequence $\sigma$, with distortion of $\mu$, over $\mu$ machines such that the competitive ratio of $ON$ is at least $\mu/2$.

We start releasing jobs, so finally $\sigma$ contains (at most) $\mu^{2}$ jobs.
By time $t$, $0 \leq t \leq \mu - 1$, $\mu$ jobs are released. All jobs have $\Tilde{d_{j}}=1$ and $p_{i,j} = 1$ for each machine $1 \leq i \leq \mu$.
Note that, even if along the time interval $[0, \mu-1]$ some of these jobs were already departed, since $ON$ using the pseudo loads, $ON$ can not taking it into account. Hence, by time $\mu$, there exist (at most) $\mu^{2}$ active jobs over $\mu$ machines, meaning that $ON$ must assign at least $\mu$ jobs on one of the machines.
When that conditions occurs, we stop the jobs sequence and we do not release any more jobs.

We can define the real durations of each of jobs as we wish (with limitation of $1 \leq d_j \leq \mu$), since either way $ON$  will act the same on the jobs sequence.
Therefore, if we choose to $d_j = \mu$ for all the jobs which $ON$ assigned to the same machine, and $d_j = 1$ otherwise, we have $ON(\sigma) = \mu$. It is also obvious that the optimal offline algorithm will assign each of these jobs to other machine, and in total $OPT(\sigma) \leq 2$, yielding the competitive ratio of $ON$ is at least $\mu / 2$.
\end{proof}

\section{An Algorithm for the \texorpdfstring{$\ell_{\infty}$-norm}{}}

Recall that $\ell_{\infty}$-norm is equivalent up to a constant factor to the $\ell_{\log m}$-norm. Hence, the results for $\ell_{\infty}$-norm can be derived from the the Load Balancing Algorithm (\ref{lp_norm_algorithm}) using $p=\log m$. In this Section, we present an alternative algorithm for this particular case. Our approach builds upon the best-known result for the problem of temporary tasks with exact durations, as outlined in \cite{AKPPW97}.

The algorithm for the $\ell_{\infty}$-norm shares the same assumptions as the algorithm for the general $\ell_p$-norm (\ref{all-assumptions}), but it also includes an additional assumption (\ref{opt_assumption}) that is crucial for its execution. For convenience, let us restate all the assumptions below:

\begin{enumerate}
    \item The parameters $\mu_{1}$ and $\tilde{D}$ are known to the algorithm in advance.
    \item All jobs arrive in the time range 0 up to $\mu_1 \tilde{D}$.
    \item \label{opt_assumption}An upper bound $OPT(\sigma) \leq \Lambda \leq 2OPT(\sigma)$ is known to the algorithm in advance.
\end{enumerate}

As in the general $\ell_p$-norm case, we give an algorithm which takes into account these assumptions. The removal of the first and second  assumptions is done exactly as explained in Subsection \ref{remove-assumptions}. In order to remove assumption (\ref{opt_assumption}), we use the generic ''doubling technique", i.e. initiating the optimal value by the first job load vector size, and then we calculate at each time the current optimal value so far, and if it grows twice more than the  optimal value, we run a new copy of the algorithm.

\paragraph*{The algorithm}\label{linf_norm_algorithm} Firstly we define for every $x$, $\bar{x} \triangleq \frac{x}{\Lambda}$. Additionally, we introduce the variable $a$ defined as $a \triangleq 1 + \frac{1}{2\mu}$. Upon arrival of a new job $j$ the algorithm allocates the job to the machine $i$ such that:
\begin{align}
    i =  \arg \min_{z} \sum_{t=t_j+1}^{t_j+\lfloor \mu_{1} \cdot \Tilde{d_j} \rfloor} \left[a^{\bar{\Tilde{\ell}}_{z,j-1}(t) + \bar{p}_{z,j}} - a^{\bar{\Tilde{\ell}}_{z,j-1}(t)}\right] \nonumber
\end{align}

\setcounter{theorem}{1}
\begin{theorempart}[Restated partly from Theorem \ref{thm:upper}]
    For the case of $\ell_\infty$-norm, the algorithm above is $\bigO\left(\mu (\log m +\log (\mu \tilde{D}))\right)$-competitive.
\end{theorempart}

\begin{proof}

For the purpose of our proof, we use the following Lemma:
\begin{lemma}\label{a^x lemma} \cite{AAFPW97}
    For every $a > 1$ and $0 \leq x \leq 1$ it holds that $a^{x} \leq (a-1)x$. 
\end{lemma}

We have for each job $j$:
{\allowdisplaybreaks 
\begin{align} \sum_{t=t_j+1}^{t_j+\lfloor \mu_{1} \cdot \Tilde{d_{j}} \rfloor} &\left[a^{\bar{\Tilde{\ell}}_{i,j-1}(t) + \bar{p}_{i,j}} - a^{\bar{\Tilde{\ell}}_{i,j-1}(t)}\right]
    \leq \sum_{t=t_j+1}^{t_j+\lfloor \mu_{1} \cdot \Tilde{d_{j}} \rfloor} \left[a^{\bar{\Tilde{\ell}}_{i^{*},j-1}(t) + \bar{p}_{i^{*},j}} - a^{\bar{\Tilde{\ell}}_{i^{*},j-1}(t)}\right] \nonumber \\
    &= \sum_{t=t_j+1}^{t_j+\lfloor \mu_{1} \cdot \Tilde{d_{j}} \rfloor} a^{\bar{\Tilde{\ell}}_{i^{*},j-1}(t)} \cdot ( a^{\bar{p}_{i^{*},j}} - 1)\nonumber \\
    &\leq \mu \cdot \sum_{t=t_j+1}^{t_j+\lceil \frac{\mu_1 \cdot \tilde{d_j}}{\mu} \rceil } a^{\bar{\Tilde{\ell}}_{i^{*},j-1}(t)} \cdot ( a^{\bar{p}_{i^{*},j}} - 1)\label{ineq41} \\
    &\leq \mu \cdot \sum_{t=t_j+1}^{t_j+d_j} a^{\bar{\Tilde{\ell}}_{i^{*},j-1}(t)} \cdot ( a^{\bar{p}_{i^{*},j}} - 1)\label{ineq42} \\
    &\leq \mu \cdot \sum_{t=t_j+1}^{t_j+d_j} a^{\bar{\Tilde{\ell}}_{i^{*}}(t)} \cdot ( a -1) \cdot \bar{p}_{i^{*},j} \label{ineq43} \\
    &= \mu(a-1) \cdot \sum_{t=t_j+1}^{t_j+d_j} a^{\bar{\Tilde{\ell}}_{i^{*}}(t)} \cdot \bar{p}_{i^{*},j}
    = \frac{1}{2} \cdot \sum_{t=t_j+1}^{t_j+d_j} a^{\bar{\Tilde{\ell}}_{i^{*}}(t)} \cdot \bar{p}_{i^{*},j} \nonumber
\end{align}}

Inequality \eqref{ineq41} follows since for any $i^{*},j$, $\tilde{\ell}_{i^{*},j}(t) \geq  \tilde{\ell}_{i^{*},j}(t+1)$. Therefore $\mu$ times the sum of the first $\lceil \frac{\lfloor \mu_1 \cdot\tilde{d}_j \rfloor}{\mu} \rceil$ terms is at least as the sum on the $\lfloor \mu_1 \cdot\tilde{d}_j \rfloor$ terms. To complete, we use that $\frac{\lfloor \mu_1 \cdot\tilde{d}_j \rfloor}{\mu} \leq \frac{\mu_1 \cdot\tilde{d}_j}{\mu}$.
Inequality \eqref{ineq42} follows since it holds $\lceil \frac{\mu_1 \cdot \tilde{d_j}}{\mu} \rceil = \lceil \frac{\tilde{d_j}}{\mu_{2}} \rceil \leq  d_{j}$.
Inequality \eqref{ineq43} follows by using Lemma \ref{a^x lemma} and the inequality $\bar{\Tilde{\ell}}_{i^{*},j-1}(t) \leq \bar{\Tilde{\ell}}_{i^{*}}(t)$.

Summing the inequalities for each $j$, we have on the $LHS$:
\begin{align}
    \sum_{j} &\sum_{t=t_j+1}^{t_j+\lfloor \mu_1 \cdot\tilde{d}_j \rfloor} \left[a^{\bar{\Tilde{\ell}}_{i,j-1}(t) + \bar{p}_{i,j}} - a^{\bar{\Tilde{\ell}}_{i,j-1}(t)}\right] \nonumber \\
    &= \sum_{i=1}^{m} \sum_{t=1}^{2\mu_1 \tilde{D}} \sum_{j | t \in [t_j+1, t_j + \lfloor \mu_1 \cdot\tilde{d}_j \rfloor]}  \left[a^{\bar{\Tilde{\ell}}_{i,j-1}(t) + \bar{p}_{i,j}\cdot y_{i,j}} - a^{\bar{\Tilde{\ell}}_{i,j-1}(t)}\right] \nonumber
    \\ &= \sum_{i=1}^{m} \sum_{t=1}^{2\mu_1 \tilde{D}} \left[a^{\bar{\Tilde{\ell}}_{i}(t)} - a^{0}\right] \label{ineq51}
    = \sum_{i=1}^{m} \sum_{t=1}^{2\mu_1 \tilde{D}} a^{\bar{\Tilde{\ell}}_{i}(t)} - 2m\mu \tilde{D}
\end{align}

Summing the inequalities for each $j$, we have on the $RHS$:
\begin{align}
    \frac{1}{2} \sum_{j} \sum_{t=t_j+1}^{t_j+d_j} a^{\bar{\Tilde{\ell}}_{i^{*}}(t)} \cdot \bar{p}_{i^{*},j} \nonumber
    &= \frac{1}{2} \sum_{i=1}^{m} \sum_{t=1}^{2\mu_1 \tilde{D}} \sum_{j | t \in [t_j+1, t_j + d_j]} a^{\bar{\Tilde{\ell}}_{i}(t)} \cdot \bar{p}_{i,j}\cdot y^{*}_{i,j} \nonumber
    \\ &= \frac{1}{2} \sum_{i=1}^{m} \sum_{t=1}^{2\mu_1 \tilde{D}} a^{\bar{\Tilde{\ell}}_{i}(t)} \cdot \sum_{j | t \in [t_j+1, t_j + d_j]} \bar{p}_{i,j}\cdot y^{*}_{i,j}
    \leq \frac{1}{2} \sum_{i=1}^{m} \sum_{t=1}^{2\mu_1 \tilde{D}} a^{\bar{\Tilde{\ell}}_{i}(t)} \label{ineq52}
\end{align}

Equality \eqref{ineq51} follows since for every $i, t$ the inner sum is telescope.
Inequality \eqref{ineq52} follows since for every machine $i$ the load of the optimal solution by time $t$ is $\ell^{*}_{i}(t) = \sum_{j | t \in [t_j+1, t_j + d_j]} p_{i,j}\cdot y^{*}_{i,j}$. By assumption \ref{opt_assumption} it holds that $\ell^{*}_{i}(t) = \sum_{j | t \in [t_j+1, t_j + d_j]} p_{i,j}\cdot y^{*}_{i,j} \leq \Lambda$. Dividing both sides by $\Lambda$ yields $\sum_{j | t \in [t_j+1, t_j + d_j]} \bar{p}_{i,j}\cdot y^{*}_{i,j} \leq 1$.

\noindent By combining the $LHS$ and the $RHS$ and exchanging the sides, we obtain:
\begin{align*}
    \sum_{i=1}^{m} \sum_{t=1}^{2\mu_1 \tilde{D}} a^{\bar{\Tilde{\ell}}_{i}(t)} \leq 4m\mu T 
\end{align*}
Thus, for every machine $i$ and time $t$ it holds:
\begin{align*}
    a^{\bar{\ell}_{i}(t)}
    \leq a^{\bar{\Tilde{\ell}}_{i}(t)}
    \leq \sum_{i=1}^{m} \sum_{t=1}^{2\mu_1 \tilde{D}} a^{\bar{\Tilde{\ell}}_{i}(t)}
    \leq 4m\mu_1 \tilde{D}
\end{align*}
By taking $\log_2$ we get:
\begin{align*}
    \bar{\ell}_{i}(t) 
    &\leq \frac{\log_{2}(4m\mu_1 \tilde{D})}{\log_{2}(a)} 
    = \frac{\log_{2}(4m\mu_1 \tilde{D})}{\log_{2}(1 + \frac{1}{2\mu})}
    \leq \frac{\log_{2}(4m\mu_1 \tilde{D})}{\frac{1}{2\mu}} \\
    &= 2\mu \log_{2}(4m\mu_1 \tilde{D}) = \bigO(\mu(\log m + \log (\mu \tilde{D})))
\end{align*}
Where we used that the function $f(x)=\log_{2}(1+x)$ is concave and it holds $f(0)=g(0)$, $f(1)=g(1)$ for the linear function $g(x)=x$, thus $f(x) \geq g(x)$ for every $0 \leq x \leq 1$, specifically for $x=\frac{1}{2\lambda}$.

\end{proof}

\section{An Algorithm for the Routing Problem}\label{appendix-routing}
In this Section, we design an online algorithm for the routing problem of temporary tasks with predictions in the $\ell_p$-norm, such that for every input series $\sigma$ with distortion $\mu$ it holds:
\[\max_{t}\|\ell(t)\|_p \leq \bigO\left(\mu (p+ \log (\mu \tilde{D}))\right) \cdot \max_{t}\|\ell^*(t)\|_p, \]

The algorithm and the proof share similarities with those used for the load balancing problem (in Subsection \ref{sec:lp_norm_upper}), with one notable distinction: in the routing problem, each job can impact the load on multiple edges rather than just one machine. To accommodate this difference, we have modified the notations (in Section \ref{sec:preliminaries}), adjusted the algorithm, and refined the proof to better suit the requirements of the routing problem.
We use the exact same assumptions which were used in Subsection \ref{sec:lp_norm_upper}.
The assumptions are the following:
\begin{enumerate}
    \item The parameters $\mu_{1}$ and $\tilde{D}$ are known to the algorithm in advance.
    \item All jobs arrive in the time range 0 up to $\mu_1 \tilde{D}$.
\end{enumerate}

\noindent The removal of the assumptions is done exactly as explained in Subsection \ref{remove-assumptions}.

\paragraph*{The algorithm}\label{lp_norm_algorithm-routing} Let $\Tilde{T} \triangleq 2\mu_1 \tilde{D}$. Upon arrival of a new job $j=(s_j, t_j, \{p_{e,j}\}_{e \in E},\tilde{d_j})$ the algorithm allocates the job to a route $W_j$ from $s_j$ to $t_j$ such that:
\begin{align}
    W_j =  \arg \min_{W'} \sum_{t=a_j + 1}^{a_j+\lfloor \mu_{1} \cdot \Tilde{d_{j}} \rfloor} \left[(\sum_{e \in E} (\Tilde{\ell}_{e,j-1}(t) + p_{e,j}\cdot \mathds{1}_{e \in W'})^{p})^{\frac{p+\log \Tilde{T}}{p}} - (\sum_{e \in E} \Tilde{\ell}_{e,j-1}(t)^{p})^{\frac{p+\log \Tilde{T}}{p}}\right] \nonumber
\end{align}
Out of all possible routes $W'$ from $s_j$ to $t_j$ in the graph.

We prove that the algorithm is $\bigO\left(\mu (p+\log (\mu \tilde{D}))\right)$-competitive for any $1 \leq p \leq \log m$.

\begin{proof}
    
For a given $j$, let $\Tilde{L}_{j-1} =  (\sum_{e \in E} \Tilde{\ell}_{e,j-1}(t)^{p})^{\frac{p+\log \Tilde{T}}{p}}$.
We have for each $j$:
{\allowdisplaybreaks
\begin{align}
    &\sum_{t=1}^{\Tilde{T}} (\sum_{e \in E} \Tilde{\ell}_{e,j}(t)^{p})^{\frac{p+\log \Tilde{T}}{p}} - \sum_{t=1}^{\Tilde{T}} (\sum_{e \in E} \Tilde{\ell}_{e,j-1}(t)^{p})^{\frac{p+\log \Tilde{T}}{p}} \nonumber\\
     &= \sum_{t=a_j + 1}^{a_j+\lfloor \mu_{1} \cdot \Tilde{d_{j}} \rfloor} (\sum_{e \in E} \Tilde{\ell}_{e,j}(t)^{p})^{\frac{p+\log \Tilde{T}}{p}} - \Tilde{L}_{j-1} \nonumber\\
    &= \sum_{t=a_j + 1}^{a_j+\lfloor \mu_{1} \cdot \Tilde{d_{j}} \rfloor} \left[(\sum_{e \in E} (\Tilde{\ell}_{e,j-1}(t) + p_{e,j}\cdot y_{e,j})^{p})^{\frac{p+\log \Tilde{T}}{p}} - \Tilde{L}_{j-1}\right] \label{routing-ineq81}\\
    &\leq \sum_{t=a_j + 1}^{a_j+\lfloor \mu_{1} \cdot \Tilde{d_{j}} \rfloor} \left[(\sum_{e \in E} (\Tilde{\ell}_{e,j-1}(t) + p_{e,j}\cdot y^{*}_{e,j})^{p})^{\frac{p+\log \Tilde{T}}{p}} - \Tilde{L}_{j-1}\right] \label{routing-ineq82}\\
    &\leq \mu \cdot \sum_{t=a_j + 1}^{a_j+ \lceil \frac{\mu_1 \cdot \tilde{d_j}}{\mu} \rceil } \left[(\sum_{e \in E} (\Tilde{\ell}_{e,j-1}(t) + p_{e,j}\cdot y^{*}_{e,j})^{p})^{\frac{p+\log \Tilde{T}}{p}} - \Tilde{L}_{j-1}\right] \label{routing-ineq86}\\
    &\leq \mu \cdot \sum_{t=a_j + 1}^{a_j+ d_{j}} \left[(\sum_{e \in E} (\Tilde{\ell}_{e,j-1}(t) + p_{e,j}\cdot y^{*}_{e,j})^{p})^{\frac{p+\log \Tilde{T}}{p}} - \Tilde{L}_{j-1}\right] \label{routing-ineq83}\\
    &\leq \mu \cdot \sum_{t=a_j + 1}^{a_j+ d_{j}} \left[(\sum_{e \in E} (\Tilde{\ell}_{e}(t) + \ell^{*}_{e,j-1}(t) + p_{e,j}\cdot y^{*}_{e,j})^{p})^{\frac{p+\log \Tilde{T}}{p}} - (\sum_{e \in E} (\Tilde{\ell}_{e}(t) + \ell^{*}_{e,j-1}(t))^{p})^{\frac{p+\log \Tilde{T}}{p}}\right] \label{routing-ineq84}\\
    &= \mu \cdot \sum_{t=a_j + 1}^{a_j+ d_{j}} \left[(\sum_{e \in E} (\Tilde{\ell}_{e}(t) + \ell^{*}_{e,j}(t))^{p})^{\frac{p+\log \Tilde{T}}{p}} - (\sum_{e \in E} (\Tilde{\ell}_{e}(t) + \ell^{*}_{e,j-1}(t))^{p})^{\frac{p+\log \Tilde{T}}{p}}\right] \nonumber\\
    &= \mu \cdot \sum_{t=1}^{\Tilde{T}} \left[(\sum_{e \in E} (\Tilde{\ell}_{e}(t) + \ell^{*}_{e,j}(t))^{p})^{\frac{p+\log \Tilde{T}}{p}} - (\sum_{e \in E} (\Tilde{\ell}_{e}(t) + \ell^{*}_{e,j-1}(t))^{p})^{\frac{p+\log \Tilde{T}}{p}}\right] \nonumber
\end{align}}

\noindent Equality \eqref{routing-ineq81} follows by the definition of $y_{e,j}$ in Section \ref{sec:preliminaries}.
Inequality \eqref{routing-ineq82} follows by the definition of the algorithm \ref{lp_norm_algorithm-routing}.
Inequality \eqref{routing-ineq86} follows since for any $e,j$, $\tilde{\ell}_{e,j}(t) \geq  \tilde{\ell}_{e,j}(t+1)$. Therefore $\mu$ times the sum of the first $\lceil \frac{\lfloor \mu_1 \cdot\tilde{d}_j \rfloor}{\mu} \rceil$ terms is at least as the sum on the $\lfloor \mu_1 \cdot\tilde{d}_j \rfloor$ terms. To complete, we use that $\frac{\lfloor \mu_1 \cdot\tilde{d}_j \rfloor}{\mu} \leq \frac{\mu_1 \cdot\tilde{d}_j}{\mu}$.
Inequality \eqref{routing-ineq83} follows since it holds $\lceil \frac{\mu_1 \cdot \tilde{d_j}}{\mu} \rceil = \lceil \frac{\tilde{d_j}}{\mu_{2}} \rceil \leq  d_{j}$.
\\Inequality \eqref{routing-ineq84} follows since the function $f(x_{1}, x_{2}, \ldots, x_{|E|}) = (\sum_{e=1}^{|E|} (x_{e} + a_{e})^{w})^{z} - (\sum_{e=1}^{|E|} x_{e}^{w})^{z}$ is non-decreasing if for all $e \in [|E|]$ when we restrict the domain of $x_{e}$ to $[0, \infty)$, $w \geq 1$, $z \geq 1$ and $a_{e} \geq 0$, and since for every $e$ it holds $\tilde{\ell}_{e,j}(t) \leq \Tilde{\ell}_{e}(t)$, in particular $\Tilde{\ell}_{e,j}(t) \leq \Tilde{\ell}_{e}(t) + \ell^{*}_{e,j-1}(t)$.

\noindent Summing the inequality for all $j$ we get,
\begin{align*} 
\sum_{t=1}^{\Tilde{T}} \|\Tilde{\ell}(t)\|_p^{p+\log \Tilde{T}} &=\sum_{t=1}^{\Tilde{T}} (\sum_{e \in E} \Tilde{\ell}_{e}(t)^{p})^{\frac{p+\log \Tilde{T}}{p}}  \\
&\leq  \mu \cdot \sum_{t=1}^{\Tilde{T}} \left[(\sum_{e \in E} (\Tilde{\ell}_{e}(t) + \ell^{*}_{e}(t))^{p})^{\frac{p+\log \Tilde{T}}{p}} - (\sum_{e \in E} (\Tilde{\ell}_{e}(t))^{p})^{\frac{p+\log \Tilde{T}}{p}}\right] \\
&= \mu \cdot \sum_{t=1}^{\Tilde{T}} \left[\| \tilde{\ell}(t) + \ell^{*}(t)\|_p^{p + \log \Tilde{T}} - \|\Tilde{\ell}(t)\|_p^{p+\log \Tilde{T}}\right]
\end{align*}
Therefore,{\allowdisplaybreaks
\begin{align}
\lefteqn{
     (1+\mu) \cdot \sum_{t=1}^{\Tilde{T}} \|\Tilde{\ell}(t)\|_p^{p+\log \Tilde{T}}
     \leq \mu \cdot \sum_{t=1}^{\Tilde{T}} \| \tilde{\ell}(t) + \ell^{*}(t)\|_p^{p + \log \Tilde{T}} } \nonumber \\ 
     &\leq \mu \cdot \sum_{t=1}^{\Tilde{T}} \left[\|\Tilde{\ell}(t)\|_p + \|\ell^{*}(t)\|_p\right]^{p + \log \Tilde{T}} \label{routing-ineq91}
     \\ &\leq \mu \cdot \sum_{t=1}^{\Tilde{T}} \left[(1+\frac{1}{4\mu(p+\log \Tilde{T})})^{p + \log \Tilde{T}} \|\Tilde{\ell}(t)\|_p^{p + \log \Tilde{T}} + (1+ 4\mu(p+\log \Tilde{T}))^{p+\log \Tilde{T}}\|\ell^{*}(t)\|_p^{p + \log \Tilde{T}}\right] \label{routing-ineq92}
     \\ &\leq \mu \cdot \sum_{t=1}^{\Tilde{T}} \left[(1+\frac{1}{2\mu})\cdot \|\Tilde{\ell}(t)\|_p^{p + \log \Tilde{T}} + (5\mu\cdot(p+\log \Tilde{T}))^{p+\log \Tilde{T}}\cdot \|\ell^{*}(t)\|_p^{p + \log \Tilde{T}}\right] \label{routing-ineq93}
     \\ &= (\mu + \frac{1}{2})\cdot \sum_{t=1}^{\Tilde{T}} \|\Tilde{\ell}(t)\|_p^{p + \log \Tilde{T}} + (5\mu\cdot(p+\log \Tilde{T}))^{p+\log \Tilde{T}}\cdot \mu \cdot \sum_{t=1}^{\Tilde{T}} \|\ell^{*}(t)\|_p^{p + \log \Tilde{T}}\nonumber
\end{align}
}
\noindent Inequality \eqref{routing-ineq91} follows by Minkowski's Inequality \cite{B05}.
Inequality \eqref{routing-ineq92} follows as for every $x, y\geq 0, a>0$, and $w\geq 0$ it holds that $(x + y)^{w} \leq (1+a)^w\cdot x^{w} + \left(1+\frac{1}{a}\right)^{w}\cdot y^{w}$, when $x=\|\Tilde{\ell}(t)\|_p$, $y=\|\ell^{*}(t)\|_p$, $w=p+\log \Tilde{T}$ and $a=\frac{1}{4\mu w}$ for each $1\leq t \leq \Tilde{T}$ separately, as showed in footnote\textsuperscript{\ref{comb-lemma}}.
Inequality \eqref{routing-ineq93} follows since $(1+\frac{1}{4\mu w})^w\leq e^{\frac{1}{4\mu}}\leq 1+\frac{1}{2\mu}$ for $\mu \geq 1$ and $w \geq 1$.

\noindent Thus,
\begin{align*}
    \frac{1}{2} \cdot \sum_{t=1}^{\Tilde{T}} \|\Tilde{\ell}(t)\|_p^{p + \log \Tilde{T}} 
    &\leq (5\mu\cdot(p+\log \Tilde{T}))^{p+\log \Tilde{T}}\cdot \mu \cdot \sum_{t=1}^{\Tilde{T}} \|\ell^{*}(t)\|_p^{p + \log \Tilde{T}}\\
    &\leq (5\mu\cdot(p+\log \Tilde{T}))^{p+\log \Tilde{T}}\cdot \mu \cdot \Tilde{T} \cdot (\max_{t=1}^{\Tilde{T}} \|\ell^*(t)\|_p)^{p+\log \Tilde{T}}
\end{align*}
In total we have:
\begin{align}
    (\max_{t}\|\Tilde{\ell}(t)\|_p)^{p+\log \Tilde{T}}
    &= \max_{t} \|\Tilde{\ell}(t)\|_p^{p+\log \Tilde{T}}
    \leq \sum_{t=1}^{\Tilde{T}} \|\Tilde{\ell}(t)\|_p^{p+\log \Tilde{T}}  \nonumber
    \\ &\leq 2 \cdot (5\mu\cdot(p+\log \Tilde{T}))^{p+\log \Tilde{T}}\cdot \mu \cdot \Tilde{T} \cdot (\max_{t=1}^{\Tilde{T}} \|\ell^*(t)\|_p)^{p+\log \Tilde{T}} \nonumber
\end{align}
Recall that $\Tilde{\ell}(t)$ is overestimate $\ell(t)$ and $\tilde{T} = 2 \mu_1 \tilde{D}$.
By taking a power of $\frac{1}{p+ \log \Tilde{T}}$ and holding that $x^{\frac{1}{\log x}} = 2$, we have:
\begin{align}
    \max_{t=1}^{\tilde{T}} \|\ell(t)\|_p 
    &\leq \max_{t=1}^{\tilde{T}} \|\Tilde{\ell}(t)\|_p 
    \leq 2^{\frac{1}{p+\log \Tilde{T}}} \cdot (5\mu\cdot(p+\log \Tilde{T})) \cdot \mu^{\frac{1}{p+\log \Tilde{T}}} \cdot \Tilde{T}^{\frac{1}{p+\log \Tilde{T}}} \cdot \max_{t=1}^{\tilde{T}} \|\ell^*(t)\|_p \nonumber \\
    &\leq 2 \cdot (5\mu\cdot(p+\log \Tilde{T})) \cdot \mu_1^{\frac{1}{p+\log \Tilde{T}}} \cdot \mu_2^{\frac{1}{p+\log \Tilde{T}}} \cdot 2 \cdot \max_{t=1}^{\tilde{T}} \|\ell^*(t)\|_p \nonumber \\
    &\leq 20 \cdot \mu\cdot(p+\log \Tilde{T}) \cdot \mu_1^{\frac{1}{\log (2\Tilde{\mu_1})}} \cdot \mu_2^{\frac{1}{\log (2\Tilde{D})}} \cdot \max_{t=1}^{\tilde{T}} \|\ell^*(t)\|_p \nonumber \\
    &\leq C \cdot \mu (p+ \log (\mu \tilde{D})) \cdot \max_{t=1}^{\tilde{T}} \|\ell^*(t)\|_p
    = \bigO\left(\mu (p+ \log (\mu \tilde{D}))\right) \cdot \max_{t=1}^{\tilde{T}} \|\ell^*(t)\|_p \nonumber
\end{align}
where the last inequality holds since $\mu_2 \leq \tilde{D}$, and for $C \leq 80$.

\end{proof}

\end{appendices}

\end{document}